\documentclass[12pt, draftclsnofoot, onecolumn]{IEEEtran}

\usepackage{kotex}
\usepackage[pdftex]{graphicx}

\usepackage{amsfonts}
\usepackage{amssymb}
\usepackage[cmex10]{amsmath}
\usepackage{amsthm}

\usepackage{amsmath,amssymb,amscd}
\usepackage{graphicx, cite}
\usepackage{algorithm}
\usepackage{algpseudocode}
\graphicspath{ {images/} }

\newtheorem{theorem}{Theorem}
\newtheorem{lemma}{Lemma}
\newtheorem{remark}{Remark}

\begin{document}

\title{On the Feasibility of Full-duplex Large-scale MIMO Cellular Systems}

\author{\IEEEauthorblockN{Jeongwan Koh, \IEEEmembership{Student Member, IEEE}, Yeon-Geun Lim, \IEEEmembership{Student Member, IEEE}, Chan-Byoung Chae, \IEEEmembership{Senior Member, IEEE}, and Joonhyuk Kang, \IEEEmembership{Member, IEEE}} }

\maketitle

\begin{abstract}
This paper concerns the feasibility of full-duplex large-scale multiple-input-multiple-output (MIMO) cellular systems. We first propose a pilot transmission scheme and assess its performance, specifically the ergodic sum-rate. The proposed scheme -- the simultaneous pilot transmission (SPT) -- enables to reduce pilot overhead, where the pilot overhead depends on the number of antennas at the base station (BS), since the self-interference channel has to be estimated. We consider two multicell scenarios-- cooperative and non-cooperative multicell systems--,  and derive the analytic model of the ergodic achievable sum-rate for cell-boundary users. The model is derived by applying a simple linear filter, i.e., matched filter or zero-forcing filter, to the BS. In the analytic model, we also consider large-scale fading, pilot contamination, transmitter noise and receiver distortion. Exploiting the derived analytic model, the feasibility of full-duplex large-scale MIMO systems is shown with respect to system parameters. In the end, we confirm that our analytic model matches well the numerical results and the SPT has advantages over other pilot transmission methods.
\end{abstract}


\IEEEpeerreviewmaketitle

\section{Introduction}
Over the last few decades, researchers have developed multiple-input-multiple-out (MIMO) technologies to provide more users with higher date rates and greater reliability~\cite{Gesbert2007shifting}. The proliferation of smart devices has led to an explosive rise in demand for higher data rate \cite{itu2015}. To handle the burgeoning data traffic, researchers have tackled various issues in fifth generation (5G) wireless communication. The key objectives for the upcoming 5G are to enhance  spectral efficiency, reduce latency, and develop cost-effective energy and hardware technology \cite{Jeff14what}. The literature \cite{Boccardi14five, Osseiran2014scenarios} has introduced some promising candidates for achieving such objectives; they  include large-scale MIMO (massive MIMO), full-duplex, millimeter wave, and cloud radio access networks (CRAN). 

Attracting a great deal of attention among these has been full-duplex technology. It is able to double spectral efficiency and reduce the roundtrip latency of system by  supporting, simultaneous  downlink (DL) and uplink (UL) transmission.  For a long time, scholars have discussed the concept of full-duplex-- the notion of sharing resources such as frequency and time. Interest, though, has been renewed now that engineers can implement full-duplex thanks to advanced antenna design and radio-frequency (RF) circuit  \cite{sim2017nonlinear, kim2017asymmetric, chung2017compact, Choi2010achieving, Bharadia2013full, chung2015prototyping}.  The most significant hurdle in full-duplex is coming up with a way to cancel the self-interference (SI) that occurs at the receiver (Rx). Such interference is caused by the signal coming from the transmitter (Tx) of the same full-duplex node. Indeed, that signal's power is much greater than the received signal power of UL users. In typical microcells (with up to 2km range), for example,  to suppress the SI to the noise floor (-90dBm), one needs approximately 125dB of cancellation.   The authors in \cite{Choi2010achieving}, \cite{Bharadia2013full} recently showed that SI can be mitigated by analog and digital cancellation in Zigbee and WiFi systems. In addition to SI, there exists user-user interference. In a single-cell full-duplex system, this type of interference occurs at the received signal of DL users and is caused by the UL signal from adjacent UL users.  Since the user-user interference cannot be mitigated by the transmit beamforming of BS, the authors in \cite{Karakus2015opportunistic} proposed a type of simple opportunistic joint UL-DL scheduling based on multiuser diversity. 

In contrast to a single-cell full-duplex system, a multicell full-duplex system must contend with 
base-station (BS)-BS interference; also user-user interference worsens due to the presence of UL users of the adjacent cell. BS-BS interference occurs at the BS's UL signal owing to the adjacent BS's DL signal. In this context, the author in \cite{Lee2015hybrid} derived the throughput by accounting for AP spatial density, SI cancellation capability, and Tx power of APs and users in the multicell full-duplex system.  For  user-user and  BS-BS interference, however, the authors either assumed those to be zero or approximated them at a simple certain value. In \cite{goyal2015full}, the authors considered this interference in a practical manner by proposing the scheduling method based on the greedy algorithm and  geometric programming. Researchers may need to reconsider, however, the assumption regarding the centralized scheduler, which can access all global system information such as the channel between the DL and  UL users. In this context, the combination of a full-duplex cellular system and CRAN seems to be very attractive for handling  user-user and BS-BS interference; such a combination would deal with the interference by means of centralized scheduling and cooperation of BSs. The authors in \cite{simeone2014full} used an information theoretic viewpoint based on the Wyner channel model to show the potential of a full-duplex system in CRAN.  In return for cooperation of the BSs, however,  quantization noise occurs over the fron-thaul connected between the central unit (CU) and BSs due to limited front-haul capacity. 

It seems inevitable that the industry will exploit large-scale MIMO at the BS side in order to provide increased spectral efficiency or support many users. Indeed, the target for 5G is a 1000-fold increase in spectral efficiency. Moreover, if the BS is able to perform full-duplex, this guarantees highly increased spectral efficiency while reducing roundtrip latency \cite{yin2013full, lim2015performance2}. Another advantage of the full-duplex large-scale MIMO system is that, as the number of antennas increases, the system reduces the average SI power by scaling down the transmit power per antenna. This reduction occurs because the average SI power of each BS antenna  depends  not on the number of transmit antennas but only on the total transmit power and the SI channel gain  \cite{yin2013full, Rusek13SPMAG, Ngo13TCOM}. However, the performance of a full-duplex large-scale MIMO system is limited by a critical pilot overhead problem. Typically, in a time-division duplex (TDD)-based half-duplex large-scale MIMO system, we can exploit the UL pilot, which  depends only on the number of DL users and their antennas to estimate the DL channel by means of the property of channel reciprocity \cite{Rusek13SPMAG}. As a result, we are able to retain the pilot overhead even if we increase the number of transmit antennas. In contrast, exploiting the DL pilots in the full-duplex system is inevitable, since the DL pilots depend on the number of transmit antennas at BS to estimate the SI channel \cite{day2012full}. Moreover, since a full-duplex enabled BS transmits and receives signals simultaneously, we need pilots to estimate both DL and UL channels within each coherence time. A full-duplex system, in other words, requires more than twice the pilot length (at least) compared with  a half-duplex system as shown in Figs. \ref{fig: pilot} (a) and (b). This problem, which can seriously limit the performance of full-duplex systems, is exacerbated as the number of antennas at BS increases. 

In this paper, we approach the large-scale MIMO cellular system with two motivations, namely, to develop an overhead reduction method and to analyze the feasibility of the system under various interference types. As noted above, pilot overhead problem and investigating the effect of the interference which does not occur in conventional half-duplex system have to be tackled. The main contributions of this paper are as follows: 
%

\setlength{\leftskip}{20pt}
\setlength{\rightskip}{20pt}
$\bullet$ \emph{Proposed pilot transmission scheme-- simultaneous pilot transmission (SPT)}: In order to reduce pilot overhead, we propose the SPT depicted in Fig. \ref{fig: pilot} (c). Conceptually, the SPT transmits the UL pilot and the SI pilot which is the DL pilot for SI channel estimation simultaneously in the time domain. We first obtain the estimated channel using the minimum-mean-square-error (MMSE) channel estimation and then investigate the achievable sum-rate of SPT by comparing it to the non-simultaneous pilot transmission (nSPT). The nSPT transmits all pilots orthogonally in the time domain, as in the conventional scheme. We observe from this that SPT holds two distinct advantages -- it reduces  pilot overhead and achieves additional power gain on estimated channels (induced by the difference of pilot length between the SI and UL pilots). At the same time, however, the channel estimation performance can be degraded due to interference between pilots. Finally, we observe this trade-off with respect to various system parameters. 

$\bullet$ \emph{Derivation of analytic model of ergodic achievable sum-rate for cell-boundary users in cooperative multicell system}: We obtain the analytic model for cell-boundary users which are  bottlenecks  for both DL and UL transmission;  at the same time, we obtain the worst case for system performance \cite{lim2015performance}. 
In addition, in the multicell scenario, cell-boundary users experience more serious user-user interference from UL users of adjacent cells.  Unlike previous studies \cite{yun2016intra, goyal2014improving}, we observe that, withe the  increase in the number of antennas at BS, we can provide a better sum-rate through a full-duplex system than is possible through a half-duplex system for cell-boundary users.  To investigate the performance in various scenarios, this study considers two scenarios -- a non-cooperative system and a cooperative multicell system. To provide an accurate analytic model in practice we consider the following four things: 1) large-scale fading; 2) pilot contamination -- a critical problem in large-scale multicell systems due to the idenntical set of UL pilots for all cells \cite{Jose11TWC}; 3) Tx noise and Rx distortion induced by limited dynamic range of Tx and Rx, as these are not negligible when we consider the SI channel \cite{day2012full}; and 4)  quantization noise for fronthaul between BSs and CU in cooperative multicell system, as we have  limited fronthaul capacity \cite{kang2014joint}. 

$\bullet$ \emph{Feasibility of full-duplex large-scale MIMO system for non- and cooperative multicell system}: Using the obtained analytic model, we analyze the behavior of a full-duplex system with very large-scale antennas at BS for two multicell system scenarios. Furthermore, we obtain the conditions needed to maintain the reliable region defined as the interval that guarantees a  performance of a full-duplex system better than those of half-duplex system. 

\setlength{\leftskip}{2pt}
\setlength{\rightskip}{2pt}

This paper is organized as follows. In Section \ref{sec: sys}, we describe the system model of UL and DL transmission for two different system scenarios.  Section \ref{sec: channel estimation} addresses the operation of  the nSPT and SPT, and shows the determination of  the distribution  resulting from the MMSE channel estimation.  Section \ref{sec: sum-rate}  uses the MF and ZF methods to derive the analytic model of ergodic achievable sum-rate for two scenarios. Section \ref{sec: analysis} shows the performance analysis, and  Section \ref{sec: simul} presents the simulation results. Conclusions are drawn  in Section \ref{sec: conclusion}.

$Notation:$ $\mathbf{A}^{H}$, $\mathbf{A}^{T}$ and $\mathbf{A}^*$ denote conjugate transpose, transpose and conjugate of matrix $\mathbf{A}$, respectively.  $\mathrm{Var}[X]$ and $\mathbb{E}[X]$ respectively imply the variance and average of random variable $X$. $\mathrm{diag}(\mathbf{A})$ and $\mathrm{blkdiag}[\mathbf{A}_1, \ldots, \mathbf{A}_n]$ denotes diagonal elements of matrix $\mathbf{A}$ and a block-diagonal matrix whose diagonal elements are $\mathbf{A}_1, \ldots, \mathbf{A}_n$. For convenience, we define $\mathcal{C}(\gamma) = \log_2(1+\gamma) $, where $\gamma$ is a random variable. We denote the column vector normalization of matrix $\mathbf{F}$ as $\mathbf{F}/||\mathbf{F}_v||$.  $\mathbf{C}_{\mathbf{XY}}= \mathbb{E} \big[ (\mathbf{X}-\mathbb{E}[\mathbf{X}])(\mathbf{Y}-\mathbb{E}[\mathbf{Y}])^\dagger \big] $ denotes the covariance matrix of $\mathbf{X}$ and $\mathbf{Y}$.

\section{System model} \label{sec: sys}
We consider the system model for two multicell scenarios -- the non-cooperative multicell and  the cooperative multicell systems. In both, we assume $N$ cells, all of which consist of a full-duplex BS with $M_t / M_r$ RF chains for transmitter (Tx) / receiver (Rx) and $K_\text{DL} / K_\text{UL}$ half-duplex users with single antenna for DL / UL transmission for each cell. The Tx and Rx RF chains are independent of each other. We assume  Rayleigh fading channel for all channels.  A line-of-sight (LoS) element is presented explicitly for the BS-BS and SI channels; nonetheless, we can still assume  Rayleigh fading channel since, prior to digital signal processing, analog cancellation can greatly mitigate the LoS component on the Rx side. All channels defined in this paper are expressed as $\mathbf{G}^\text{x}=\mathbf{H}^{\text{x}}(\mathbf{D}^{\text{x}})^{\frac{1}{2}}$, 
where $\mathbf{H}^\text{x}$ includes fading coefficients which follow zero mean and unit variance, and $\text{x}\in\{d, u, \text{BS}, \text{UE}\}$ indicates the DL and uplink channels,  the channel between BSs and the channel between UL users and DL users. $(\mathbf{D}^{\text{x}})^{\frac{1}{2}}$ denotes a diagonal matrix with  $[\mathbf{D}^\text{x}]_{kk}=\rho_{k}^2$ for the $k$th diagonal element which represents the geometric attenuation. We assume that $\mathbf{D}^{\text{x}}$ is known in advance, as it changes very slowly with time. Accordingly, instead of estimating $\mathbf{G}^\text{x}$, we focus on estimating $\mathbf{H}^\text{x}$. In order to mitigate the SI at each BS,  we adopt a time-domain cancellation that directly subtracts the estimated SI channel from the received signal. In this context, the variance of the estimation error of the SI channel becomes the power of the residual SI. Moreover, considering hardware impairment, we reflect Tx noise and Rx distortion, which need to be considered in signals   transmitted or received signal over the  SI channel \cite{day2012full}.  Since the distance between BS and the users is sufficiently large, we ignore Tx noise and Rx distortion for all channels, other than the SI channel. In the following subsections, we  address details of the  two different scenarios.

\subsection{System Scenario 1: Non-Cooperative Multicell System}
Here, based on local estimated channel state information (CSI), each BS produces a precoder for DL transmission and a detection filter for UL transmission. Since no cooperation exists between any of the BSs, intercell interference is only slightly mitigated. 

\subsubsection{Downlink Transmission} 
For DL channels between BS of cell $i$ and $K_{\text{DL}}$ users of cell $j$, we define $(\mathbf{G}_{i,j}^d)^T = ({\mathbf{H}}_{i,j}^d (\mathbf{D}_{i,j}^d)^{\frac{1}{2}})^T \in \mathbb{C}^{ M_t\times K_\text{DL}} = \big[\mathbf{g}^d_{ij,1} \ldots \mathbf{g}^d_{ij,K_\text{DL}}\big]$, where $[\mathbf{D}_{i,j}^d]_{kk}=(\rho_{ij,k}^d)^2$. We define the $k^\text{th}$ column of the estimated DL channel $(\hat{\mathbf{G}}_{i,j}^d)^T$ and the estimation error  $(\mathbf{\Upsilon}_{i,j}^d)^T$ as $\hat{\mathbf{g}}_{ij,k}^d \sim \mathcal{CN}(0, \hat{\rho}_{ij.k}^2 \mathbf{I}_{M_t})$ and $\mathbf{\epsilon}_{ij,k}^d \sim \mathcal{CN}(0, \bar{\rho}_{ij.k}^2 \mathbf{I}_{M_t})$, where $\mathbf{G}_{i,j}^d = \hat{\mathbf{G}}_{i,j}^d + \mathbf{\Upsilon}_{i,j}^d$. $\hat{\rho}_{ij.k}^2$ and $\bar{\rho}_{ij.k}^2$ are determined by pilot transmission method which will be discussed in Sec. \ref{sec: channel estimation}. We define the precoder of cell $i$, $\mathbf{F}_i=[\mathbf{f}_{i,1} \ldots \mathbf{f}_{i,K_{DL}}] \in \mathbb{C}^{M_t \times K_\text{DL}}$, and the channel between DL users of cell $i$ and UL users of cell $j$, $(\mathbf{G}_{i,j}^\text{UE})^{\text{T}}=({\mathbf{H}}_{i,j}^\text{UE}
\mathbf{D}_{i,j}^\text{UE})^\mathrm{T}=[\mathbf{g}_{ij,1}^\text{UE} \ldots \mathbf{g}_{ij,K_{DL}}^\text{UE}] \in \mathbb{C}^{ K_\text{UL} \times K_\text{DL}}$, where $[\mathbf{D}_{i,j}^\text{UE}]_{kk}=(\rho_{ij,k}^\text{UE})^2$. The received signal of DL user $k$ at cell $i$ is 
\begin{eqnarray} \label{eq:y_NoCRAN}
\hspace{-1cm}y_{d,ik}=\underbrace{\sqrt{P_d}(\mathbf{\hat{g}}_{ii,k}^d)^T \mathbf{f}_{i,k} s_{i,k}^d} _{\text{desired\; signal}} +  \underbrace{\sqrt{P_d}\sum_{\ell=1,\ell\neq k}^{K_\text{DL}} (\mathbf{g}_{ii,k}^d)^T \mathbf{f}_{i,\ell} s_{i,\ell}^d}_{\text{intra-cell \; interference}} +  \underbrace{ \sqrt{P_d}\sum_{j=1, j\neq i}^N (\mathbf{g}_{ij,k}^d)^T \mathbf{F}_j \mathbf{s}_j^d }_{\text{inter-cell \; interference}} \nonumber\\ + 
\underbrace{\sqrt{P_u}\sum_{j=1}^N (\mathbf{g}_{ij,k}^\text{UE})^T \mathbf{s}_j^u}_{\text{UE-UE \; interference}} + \underbrace{\sqrt{P_d}(\mathbf{\epsilon}_{ii,k}^d)^T \mathbf{f}_{i,k} s_{i,k}^d }_{\text{estimation\; error}} + \underbrace{ n_{di,k} }_{\text{noise}}, \hspace{-1.5cm} 
\end{eqnarray}
where $\mathbf{s}_j^d \in \mathbb{C}^{K_\text{DL} \times 1}$ and $\mathbf{s}_j^u \in \mathbb{C}^{K_\text{UL} \times 1}$ denote the  DL and UL transmitted symbols of cell $j$. Without loss of generality, we assume that $ \mathbf{s}_j^d (\mathbf{s}_j^d)^H = \mathbf{I}_{K_\text{DL}}$ and $\mathbf{s}_j^u (\mathbf{s}_j^u)^H = \mathbf{I}_{K_\text{UL}}$, and the Gaussian noise $n_{i,k}^d \sim \mathcal{CN}(0,n_0)$. In (\ref{eq:y_NoCRAN}), $P_d$ and $P_u$ denote downlink transmit power per antenna and uplink transmit power per user. We note that the UE-UE interference is induced by adjacent UL users.  
 
\subsubsection{Uplink Transmission}
In a  manner similar to that of DL transmission, for UL channels between the BSs of cell $i$ and $K_{\text{UL}}$ users of cell $j$, we define $\mathbf{G}_{i,j}^u = {\mathbf{H}}_{i,j}^u (\mathbf{D}_{i,j}^u)^{\frac{1}{2}} \in \mathbb{C}^{M_r \times K_\text{UL}} = \big[\mathbf{g}^u_{ij,1} \ldots \mathbf{g}^u_{ij,K_\text{UL}}\big]$, where $[\mathbf{D}_{i,j}^u]_{kk}=(\rho_{ij,k}^u)^2$ . We also define the channels between BSs of cell $i$ and cell $j$ as $\mathbf{G}_{i,j}^\text{BS}= {\mathbf{H}}_{i,j}^\text{BS} (\mathbf{D}_{i,j}^\text{BS})^{\frac{1}{2}}\in \mathbb{C}^{M_r \times M_t}$, where $\mathbf{D}_{i,j}^\text{BS} = (\rho_{ij}^\text{BS})^2 \mathbf{I}_{M_t}$ for $i\neq j$. When it comes to the SI channel, i.e., the $i=j$ case, $\mathbf{D}^{\text{BS}}_{i,i}$ is a symmetric matrix whose elements are defined as $[\mathbf{D}^{\text{BS}}_{i,i}]_{\ell m}= (\rho_{ii,\ell m}^\text{BS})^2$ which reflects the distance between each antenna.  The detection filter of cell $i$ is defined as $\mathbf{W}_i^{{T}}=[\mathbf{w}_{i,1} \ldots \mathbf{w}_{i,K_\text{UL}}] \in \mathbb{C}^{M_r \times K_\text{UL}}$.  The received signal for user $k$ of cell $i$ at BS is expressed as
\begin{eqnarray}\label{eq:y_NoCRAN_ul}
\hspace{-0.8cm} y_{u,ik}=\hspace{-0.5cm}&&\underbrace{\sqrt{P_u}(\mathbf{w}_{i,k})^T \hat{\mathbf{g}}_{ii,k}^u s_{i,k}^u}_{\text{desired\; signal}} + \underbrace{\sqrt{P_u}(\mathbf{w}_{i,k})^T \!\!\!\! \sum_{l=1,l\neq k}^{K_\text{UL}} \!\!\!\! \mathbf{g}_{ii,l}^u s_{i,l}^u}_{\text{intra-cell\; interference}} 
 + \underbrace{\sqrt{P_u}(\mathbf{w}_{i,k})^T \!\!\!\! \sum_{j=1, j\neq i}^N  \!\!\!\! \mathbf{G}_{ij}^u \mathbf{s}_j }_{\text{inter-cell\; interference}}   \nonumber\\ &&\hspace{-1.1cm}+ \underbrace{\sqrt{P_d}(\mathbf{w}_{i,k})^T (\mathbf{G}_{ii}^\text{BS}-\hat{\mathbf{G}}_{ii}^\text{BS}) \mathbf{F}_i \mathbf{s}_i^d }_{\text{residual self\; interference}} + \underbrace{\sqrt{P_d}(\mathbf{w}_{i,k})^T \!\!\!\! \sum_{j=1, j\neq i}^N \!\!\!\! \mathbf{G}_{ij}^\text{BS} \mathbf{F}_j \mathbf{s}_j^d  }_{\text{BS-BS\; interference}}+ \underbrace{\sqrt{P_u}(\mathbf{w}_{i,k})^T {\mathbf{\epsilon}}_{ii,k}^u s_{i,k}^u}_{\text{estimation\; error} }
 \nonumber\\ &&\hspace{-1cm}+ \underbrace{\sqrt{P_d}(\mathbf{w}_{i,k})^T\mathbf{G}_{ii}^\text{BS} \mathbf{F}_i \mathbf{\psi}_{i}}_{\text{Tx\; noise}} + \underbrace{(\mathbf{w}_{i,k})^T \mathbf{\delta}_{i}}_{\text{Rx\; distortion}} + \underbrace{(\mathbf{w}_{i,k})^T \mathbf{n}_{ui,k}}_{\text{noise}}.  
\end{eqnarray}
Based on \cite{day2012full}, Tx noise and Rx distortion occurring at cell $i$ are modeled as $\mathbf{\psi}_i \in \mathbb{C}^{K_\text{DL} \times 1}$ and $ \mathbf{\delta}_i \in \mathbb{C}^{M_r \times 1} $, respectively, where the distribution of each is given by $\mathbf{\psi}_{i} \sim \mathcal{CN}\big(0, \alpha \text{diag}(\mathbf{s}_{i} \mathbf{s}_{i}^H) \big)$ and $\mathbf{\delta}_{i} \sim \mathcal{CN}\big(0, \beta \text{diag}( \mathbf{y}_{i}^\text{SI} (\mathbf{y}_{i}^\text{SI}) ^H) \big)$. 
Typically, $\alpha \ll 1$ and $\beta \ll 1$ \cite{suzuki2008transmitter}. $ \mathbf{y}_i^\text{SI} = \sqrt{P_d} \mathbf{G}_{ii}^\text{BS} \mathbf{F}_i (\mathbf{s}_i^d + \mathbf{\psi}_i) + \mathbf{n}_{ui,k}$ denotes an undistorted signal received over the SI channel.
\begin{figure}[t]
\centering
\includegraphics[width=17cm, height=5.3cm]{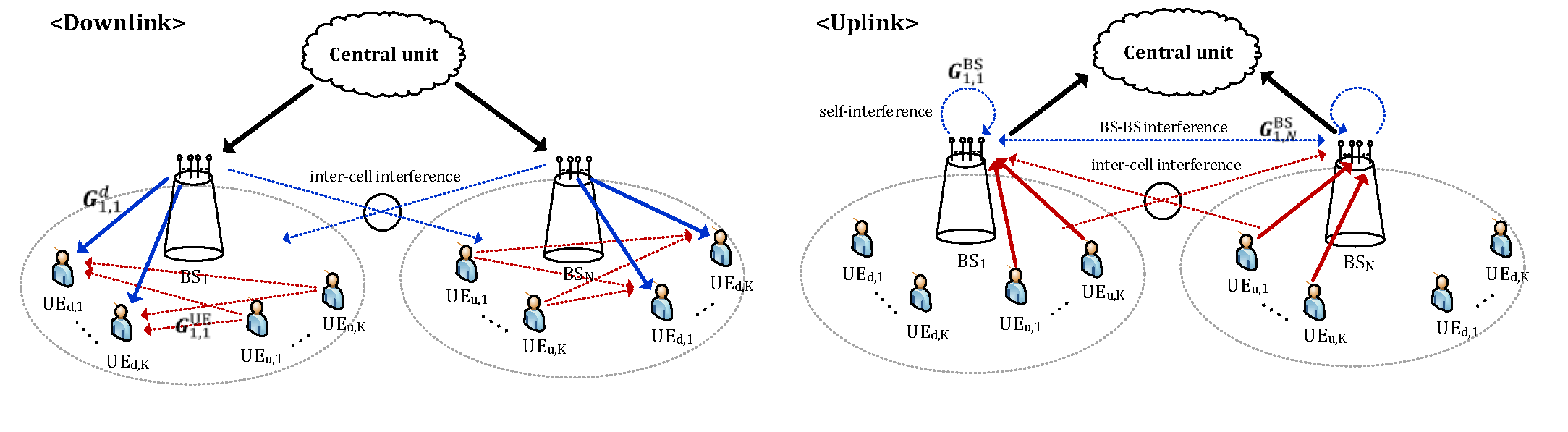}
\caption{DL and UL transmission in a cooperative multicell system with a central unit, full-duplex  BS, and half-duplex users. Here, UE stands for user equipment. Though we describe DL and UL transmission separately for convenience, both transmission operate simultaneously. } \label{fig:sys}
\end{figure}

\subsection{System Scenario 2: Cooperative Multicell System }
As depicted in Fig. \ref{fig:sys}, we have a CU that is connected to each BS via front-haul with limited capacity $C_d / C_u$ for DL / UL transmission. Based on the collected global CSI, the CU produces a precoder and a detection filter based on the collected global CSI.  As a result, the system enables  mitigation of the intercell interference that stems from BS cooperation. 

\subsubsection{Downlink Transmission}
We define all DL channels of the system between BSs and DL users as $\mathbf{G}^d=\mathbf{H}^d (\mathbf{D}^d)^{\frac{1}{2}} \in \mathbb{C}^{K_\text{DL}N \times M_t N}$, and $(\mathbf{G}^d)^T = \big[\mathbf{g}_{11}^d\ldots \mathbf{g}^d_{1K_\text{DL}} \ldots \mathbf{g}_{N1}^d \ldots \mathbf{g}^d_{NK_\text{DL}} \big]$, for convenience, where $\mathbf{g}_{jk}^d \sim \mathcal{CN}(0, \mathrm{blkdiag}\big[(\rho_{1j,k}^d)^2 \mathbf{I}_{M_t} \ldots (\rho_{Nj,k}^d)^2 \mathbf{I}_{M_t}\big]$. 
We define $\mathbf{F}^\text{C}=[\mathbf{f}_{11}^\text{C} \ldots \mathbf{f}_{1K_\text{DL}}^\text{C} \ldots \mathbf{f}_{N1}^\text{C} \ldots \mathbf{f}_{NK_\text{DL}}^\text{C}] \in \mathbb{C}^{M_tN \times K_\text{DL}N}$ as a precoding matrix.  The quantization noise over DL front-haul of cell $i$ is described as $ \mathbf{q}_i^d \sim \mathcal{CN}(0,(\sigma_i^d)^2 \mathbf{I}_{M_t}),$ where $(\sigma_i^d)^2=P_s\text{E}||\mathbf{x}_i||^2/(2^{C_d}-1)$ due to the limited front-haul capacity $C_d=\log_2(1+ P_s \mathbb{E}||\mathbf{x}_i||^2/(\sigma_i^d)^2 ) $.  $P_s$ denotes the desired symbol power and   follows $P_s = P_d(1-2^{-C_d})$ due to $P_d = P_s + (\sigma_i^d)^2$ \cite{simeone2014full}. The precoded signal transmitted from CU to the $i^\text{th}$ BS is defined as $\mathbf{x}_i=\mathbf{F}_i^\text{C}\mathbf{s}^d$, where $\mathbf{F}_i^\text{C}$ is row vectors of $\mathbf{F}$ from $(M_t(i-1)+1)^\text{th}$ to $(M_ti)^\text{th}$. Then, the quantization noise for $N$ cells is  $\mathbf{q}^d \sim \mathcal{CN} \big(0, \text{blkdiag}[(\sigma_{1}^d)^2 \mathbf{I}_{M_t}, \ldots, (\sigma_{N}^d)^2 \mathbf{I}_{M_t}] \big)$.
The received signal of user $k$ at cell $i$ is  
\begin{eqnarray} \label{eq:y_CRAN}
\hspace{-1cm}y_{d,ik}^\text{C}= 
\underbrace{\sqrt{P_s} (\hat{\mathbf{g}}_{i,k}^d)^T \mathbf{f}_{i,k}^{\text{C}} \mathbf{s}^d_{i,k}}_{\text{desired\; signal}} + \underbrace{\sqrt{P_d}\!\!\!\!\!\!\!\!\!\!\!\!\!\!\!\!\!\sum_{(n,j)=(1,1),(n,j)\neq(i,k)}^{(N,K_\text{DL})}  \!\!\!\!\!\!\!\!\!\!\!\!\!\!\!\!\! (\mathbf{g}_{i,k}^d)^T \mathbf{f}_{n,j}^\text{C} \mathbf{s}^d_{n,j}}_{\text{intra,inter-cell\; interference}} + 
\underbrace{\sqrt{P_u} \sum_{j=1}^N (\mathbf{g}_{ij,k}^\text{UE})^T \mathbf{s}_j^d}_{\text{UE-UE\; interference}}\nonumber\\ + \underbrace{({\mathbf{g}}_{i,k}^d)^T \mathbf{q}^d}_{\text{quantization\; noise}} +  \underbrace{\sqrt{P_s} (\hat{\mathbf{\epsilon}}_{i,k}^d)^T \mathbf{f}_{i,k}^\text{C} s_{i,k}^d}_{\text{estimation\; error}} + \underbrace{n_{d,ik}^\text{C}}_{\text{noise}}.\hspace{-1.6cm}
\end{eqnarray}

\subsubsection{Uplink Transmission}
In a manner similar to that of the DL transmission, we define the all UL channel of the system as $\mathbf{G}^u=\mathbf{H}^u (\mathbf{D}^u)^{\frac{1}{2}} \in \mathbb{C}^{M_rN \times K_\text{UL}N}$. Based on the UL channels, CU produces detection filter defined as $(\mathbf{W}^\text{C})^T=[\mathbf{w}_{1,1}^\text{C} \ldots \mathbf{w}_{N,K}^\text{C}] \in \mathbb{C}^{MN \times K_\text{UL}N}$. The  signal received at $\text{BS}_j$ is
\begin{eqnarray}\label{eq: y_ul_i}
\mathbf{y}_j^u=\sqrt{P_u}\mathbf{G}_j^u \mathbf{s}^u + \sqrt{P_s} \mathbf{G}_j^{BS} \mathbf{F}^\text{C} \mathbf{s}^d + \sqrt{P_s}\mathbf{G}_{jj}^\text{BS}\mathbf{F}_j^\text{C}\mathbf{\psi}^\text{C}+ \mathbf{\delta}_j^\text{C}+ \mathbf{n}_j,
\end{eqnarray}
where $\mathbf{G}_j^u \in \mathbb{C}^{M_r \times K_\text{UL}N}$ and $\mathbf{G}_j^{BS} \in \mathbb{C}^{M_r \times M_tN}$ denote the row vectors from $(M_r(j-1)+1)^\text{th}$ to $(M_rj)^\text{th}$ of $\mathbf{G}^u$ and the channel between $\text{BS}_j$ and all other BSs, respectively.
In (\ref{eq:y_CRAN_ul}), $\mathbf{\delta}_j^\text{C}$ denotes the column vector from the $(M_r(j-1)+1)^\text{th}$ to the $(M_rj)^\text{th}$ of $\mathbf{\delta}^\text{C}$. The UL quantization noise of $\text{BS}_j$ is $\mathbf{q}_j^u \sim \mathcal{CN}(0,(\sigma_j^u)^2 \mathbf{I}_{M})$, where $(\sigma^u)^2=\mathbb{E}||\mathbf{y}_j||^2/(2^{C_u}-1)$. The  UL quantization noise of $N$ cells, $\mathbf{q}^u$ is also defined in a similar manner to that of the DL case. The received signal of user $k$ of cell $i$ at CU is 
\begin{eqnarray} \label{eq:y_CRAN_ul}
&&\hspace{-0.6cm}y_{u,ik}^\text{C}=\underbrace{\sqrt{P_u} (\mathbf{w}_{i,k}^{\text{C}})^T \hat{\mathbf{g}}_{i,k}^u s^u_{i,k}}_{\text{desired\; signal}} \;\; + \underbrace{\sqrt{P_u} \!\!\!\!\!\!\!\!\!\!\!\!\!\!\!\!\!\sum_{(n,j)=(1,1),(n,j)\neq(j,k)}^{(N,K_\text{UL})}\!\!\!\!\!\!\!\!\!\!\!\!\!\!\!\!\! (\mathbf{w}_{i,k}^{\text{C}})^T \mathbf{g}_{n,j}^u s^u_{n,j}}_{\text{intra,inter-cell \; interference} } + \underbrace{\sqrt{P_s} (\mathbf{w}_{i,k}^{\text{C}})^T \mathbf{G}^\text{BS}_\text{off} \mathbf{F}^\text{C} \mathbf{s}^d}_{\text{BS-BS\;interfernce}} \nonumber\\  
&&\hspace{-0.5cm}+\underbrace{\sqrt{P_u} (\mathbf{w}_{i,k}^{\text{C}})^T \mathbf{\epsilon}_{i,k}^u s^u_{i,k}}_{\text{estimaion \; error}}  
\underbrace{\sqrt{P_s} (\mathbf{w}_{i,k}^{\text{C}})^T (\mathbf{G}^\text{BS}_\text{diag} - \hat{\mathbf{G}}^\text{BS}_\text{diag}) \mathbf{F}^\text{C} \mathbf{s}^d}_{\text{residual\;SI}} +
\underbrace{(\mathbf{w}_{i,k}^{\text{C}})^T \mathbf{q}^u}_{\text{quantization\; noise}}+ \underbrace{\sqrt{P_s} (\mathbf{w}_{i,k}^{\text{C}})^T \mathbf{G}^\text{BS}_\text{diag} \mathbf{F}^\text{C} \mathbf{\psi}^\text{C}}_{\text{Tx\;noise}} \nonumber\\ 
&&\hspace{-0.5cm}+\underbrace{(\mathbf{w}_{i,k}^{\text{C}})^T \mathbf{\delta}^\text{C}}_{\text{Rx\;distortion}}+
\underbrace{(\mathbf{w}_{i,k}^{\text{C}})^T \mathbf{n}_{u,ik}^\text{C}}_{\text{noise}},
\end{eqnarray} \\
where $\mathbf{G}^\text{BS}=\mathbf{G}_\text{off}^\text{BS} + \mathbf{G}_\text{diag}^\text{BS}$. We define   $\mathbf{G}_\text{diag}^\text{BS}$ is a block-diagonal matrix whose diagonal elements are $[ \mathbf{G}_{1,1}^\text{BS} \ldots \mathbf{G}_{N,N}^\text{BS}]$. We define Tx noise and Rx distortion as $\mathbf{\psi}^\text{C} \sim \mathcal{CN}(0, \alpha \text{diag}(\mathbf{s}^d (\mathbf{s}^d)^H)$ and $\mathbf{\delta}^\text{C} \sim \mathcal{CN}(0, \beta \text{diag}(\mathbf{y}^\text{C,SI}(\mathbf{y}^\text{C,SI})^H )$, respectively, where $\alpha \ll 1,\; \beta \ll 1$, and $\mathbf{y}^\text{C,SI}=\sqrt{P_s}  \mathbf{G}^\text{BS}_\text{diag} \mathbf{F}^\text{C}(\mathbf{s}^d + \mathbf{\psi}^\text{C})  + \mathbf{n}_{u,ik}^\text{C}$ \cite{day2012full, suzuki2008transmitter}.  

\begin{figure}[t]
\centering
\includegraphics[width=16.5cm, height=9cm]{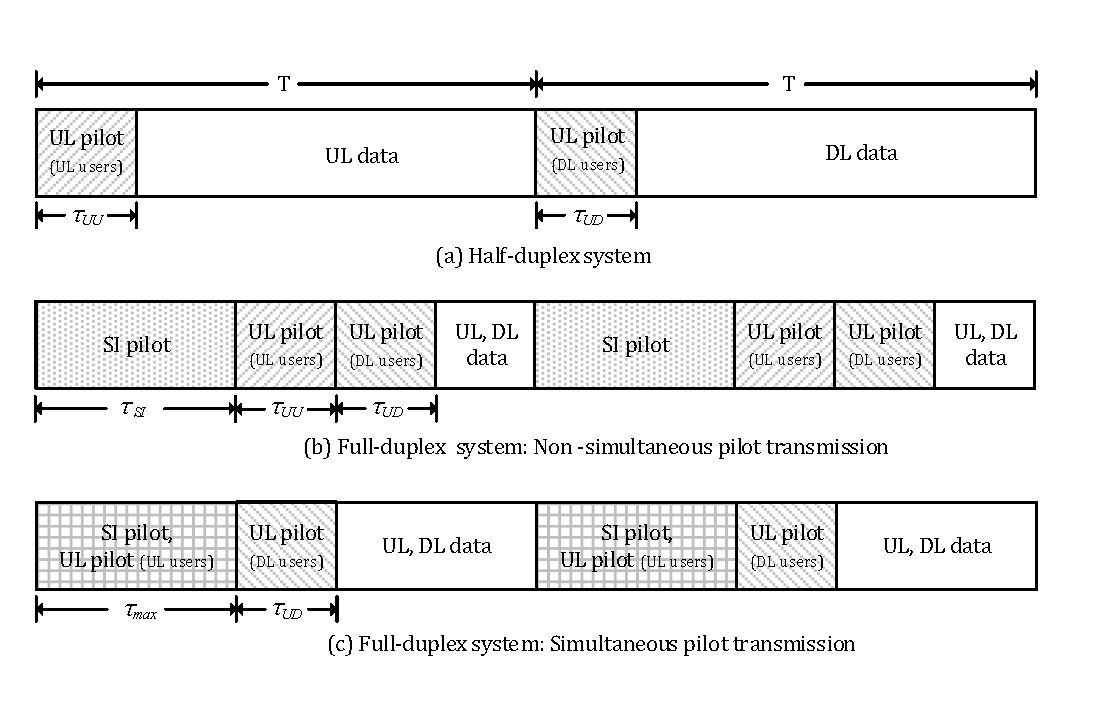}
\caption{Pilot transmission methods in half-duplex and full-duplex large-scale MIMO system. The users noted in parentheses beside the UL pilot indicates the sending location. That is, UL pilot (DL users) refers to the UL pilot sent from the DL users.     } \label{fig: pilot}
\end{figure}

\section{Proposed Pilot Transmission Scheme} \label{sec: channel estimation}
Here, we introduce two pilot transmission schemes in a TDD-based full-duplex large-scale MIMO system; in these schemes,  each BS performs channel estimation based on the received signal. One is the conventional scheme, namely nSPT;  the  other is our proposed scheme, SPT. Based on the channel reciprocity in large-scale MIMO systems, we exploit UL pilots to estimate both DL and UL channels \cite{Rusek13SPMAG, Ngo13TCOM}. We define two different UL pilots in order to distinguish their usage, where  $\mathbf{\Phi}_i^\text{UU} \in \mathbb{C}^{\tau_\text{UU} \times K_\text{UL}} \;(\tau_\text{UU} \ge K_\text{UL})$ denote the UL pilots sent from UL users to estimate UL channels which follows  $(\mathbf{\Phi}_i^\text{UU})^H \mathbf{\Phi}_i^\text{UU} = \mathbf{I}_{K_\text{UL}}$. In a similar manner, $\mathbf{\Phi}_i^\text{UD} \in \mathbb{C}^{\tau_\text{UD} \times K_\text{DL}} \;(\tau_\text{UD} \ge K_\text{DL})$ denotes the UL pilots sent from DL users to estimate the DL channels that satisfies $(\mathbf{\Phi}_i^\text{UD})^H \mathbf{\Phi}_i^\text{UD} = \mathbf{I}_{K_{DL}}$. 
We define the SI pilots to estimate the SI channels as $\mathbf{\Phi}_i^\text{SI} \in \mathbb{C}^{\tau_\text{SI} \times M_t} \; (\tau_\text{SI} \ge M_t), (\mathbf{\Phi}_i^\text{SI})^H \mathbf{\Phi}_i^\text{SI} = \mathbf{I}_{M_t}$.
Furthermore, because in large-scale MIMO, we use the same set of UL pilots in each cell, we consider the effect of pilot contamination in the channel estimation \cite{Jose11TWC}. We recall that Tx noise and Rx distortion on  signals received over SI channels are also considered. Throughout this paper, we use  MMSE channel estimation \cite{sengijpta1995fundamentals}.

\subsection{Conventional Scheme: Non-Simultaneous Pilot Transmission (nSPT)}
As shown in Fig. \ref{fig: pilot}, three different pilots are transmitted orthogonally in the time domain prior to sending the data needed to estimate the UL, DL and SI channels. Thus, the required pilot overhead is $(\tau_\text{UD} + \tau_\text{UU} + \tau_\text{SI})$. Since the UL and DL channels are estimated independently, we follows the results of channel estimation in \cite{Ngo13TCOM}. Here, we  describe only the procedures of SI channel estimation for nSPT. During SI pilot transmission, the received signal for Rx of BS at cell $i$ is 
\begin{eqnarray}
\mathbf{Y}_i= \sqrt{\tau_{SI}P_d} \mathbf{G}_{ii}^\text{BS} (\mathbf{\Phi}_i^\text{SI}+\mathbf{\Psi}_i)^\mathrm{T} + \mathbf{N}_i+\mathbf{\Delta}_i = \bar{\mathbf{Y}}_i + \mathbf{\Delta}_i, 
\end{eqnarray}
where $\mathbf{\Psi}_i \in \mathbb{C}^{\tau_\text{SI} \times M_t}$ and $\mathbf{\Delta}_i \in \mathbb{C}^{\tau_\text{SI} \times M_t}$ are Tx noise and Rx distortion. We define $(\mathbf{\Phi}_i^\text{SI})^\mathit{T}=[\mathbf{\phi}_{i1}^\text{SI} \ldots \mathbf{\phi}_{i\tau_\text{SI}}^\text{SI}]$, $ (\mathbf{\Psi}_i)^\mathit{T}=[\mathbf{\psi}_{i1} \ldots \mathbf{\psi}_{i\tau_\text{SI}}]$ and $(\mathbf{\Delta}_i)^\mathit{T}=[ \mathbf{\delta}_{i1} \ldots \mathbf{\delta}_{i\tau_\text{SI}}]$, where the $t^{th}$ column of $\mathbf{\Psi}_i$ and $\mathbf{\Delta}_i$ follow $\mathbf{\psi}_{it} \sim \mathcal{CN}\big(0, \alpha \text{diag}(\mathbf{\phi}_{it}^\text{SI} (\mathbf{\phi}_{it}^\text{SI})^H) \big)$ and $\mathbf{\delta}_{it} \sim \mathcal{CN}\big(0, \beta \text{diag}( \bar{\mathbf{y}}_{it} (\bar{\mathbf{y}}_{it}) ^H) \big)$, where
$\bar{\mathbf{y}}_{it}$ is the $t^\text{th}$ column of $\bar{\mathbf{Y}}_i$. We obtain the distribution of the estimated channels as follows: 
\begin{subequations}
\begin{eqnarray}
&&\hspace{-1.5cm}\hat{\mathbf{g}}_{ii,k}^d \sim \mathcal{CN}\left(0, \frac{\tau_{UD}P_u (\rho_{ii,k}^{d})^4}{\tau_{UD}P_u \sum_{j=1}^N (\rho_{ij,k}^{d})^2 + n_o} \mathbf{I}_{M_r}\right), \label{eq: est_nSPT_a}\\
&&\hspace{-1.5cm}\hat{\mathbf{g}}_{ii,k}^u \sim \mathcal{CN}\left(0, \frac{\tau_{UU}P_u (\rho_{ii,k}^{u})^4}{\tau_{UU}P_u \sum_{j=1}^N (\rho_{ij,k}^{u})^2 + n_o} \mathbf{I}_{M_r}\right),  \label{eq: est_nSPT_b}\\
&&\hspace{-1.5cm}\hat{{g}}_{ii,\ell m}^{BS} \sim \mathcal{CN}\left(0, \frac{\tau_{SI} P_d (\rho_{ii,\ell m}^{BS})^4 } { (1+ \beta)\{\tau_{SI} P_d (\rho_{ii,\ell m}^{BS})^2 + \alpha P_d \sum_{m=1}^{M_t}(\rho_{ii,\ell m}^{BS})^2 + n_0 \} } \right). \label{eq: est_nSPT_c}
\end{eqnarray}
\end{subequations}

\subsection{Proposed Scheme: Simultaneous Pilot Transmission (SPT)}
As illustrated in Fig. \ref{fig: pilot} (c), the main concept of the proposed scheme is that  Tx RF chains of both BS and UL users send pilots simultaneously to reduce pilot overhead. The DL users cannot simultaneously transmit pilots with Tx RF chains of BS since such pilots have to go through the channel between Tx RF chains and DL users. Thus, we redefine the pilot as $\mathbf{\Phi}_i^\text{UU} \in \mathbb{C}^{\tau_\text{max} \times K_\text{UL}}$ and $\mathbf{\Phi}_i^\text{SI} \in \mathbb{C}^{\tau_\text{max} \times M_t}$, where $\tau_\text{max}= \max (\tau_\text{SI}, \tau_\text{UU})$. The received BS signal at cell $i$ is 
\begin{eqnarray} \label{eq: Y_i_SPT}
&&\hspace{-2cm}\mathbf{Y}_i^\text{SPT}= \sqrt{\tau_\text{max}P_u} \mathbf{G}_{i,i}^u (\mathbf{\Phi}_i^\text{UU})^{T} + {\sqrt{\tau_\text{max}P_u} \sum_{j=1, j\neq i}^N \mathbf{G}_{i,j}^u (\mathbf{\Phi}_j^\text{UU})^{T}} \nonumber \\ 
&& \hspace{3.5cm}+ \sqrt{\tau_\text{max}P_d} \mathbf{G}_{ii}^\text{BS} (\mathbf{\Phi}_i^\text{SI}+\mathbf{\Psi}_i^\text{SPT})^{T} + \mathbf{N}_i + \mathbf{\Delta}_i^\text{SPT},  
\end{eqnarray}
where $\mathbf{\psi}_{it}^\text{SPT} \sim \mathcal{CN}\big(0, \alpha \text{diag}(\mathbf{\phi}_{it}^\text{SI} (\mathbf{\phi}_{it}^\text{SI})^H) \big)$ and $\mathbf{\delta}_{it}^\text{SPT} \sim \mathcal{CN}\big(0, \beta \text{diag}( \bar{\mathbf{y}}_{it}^\text{SPT} (\bar{\mathbf{y}}_{it}^\text{SPT}) ^H) \big)$.
And, $\bar{\mathbf{y}}_{it}^\text{SPT}$ is the $t^\text{th}$ column of $\bar{\mathbf{Y}}_i^\text{SPT}= \sqrt{\tau_\text{max}P_d} \mathbf{G}_{ii}^\text{BS} (\mathbf{\Phi}_i^\text{SI}+\mathbf{\Psi}_i^\text{SPT})^{T} + \mathbf{N}_i$. In order to estimate both the SI and UL channels based on (\ref{eq: Y_i_SPT}), we estimate the SI channels first because the power of the SI pilots is larger than that of the UL pilots. Then, we subtract the estimated signal from (\ref{eq: Y_i_SPT}). The resulting signal is 
\begin{eqnarray}
&&\hspace{-2.5cm}{\mathbf{Y}}_i^\text{SPT,r}= \sqrt{\tau_\text{max}P_u} \sum_{j=1}^N \mathbf{G}_{i,j}^u (\mathbf{\Phi}_j^\text{UU})^{T} + \sqrt{\tau_\text{max}P_d} (\mathbf{G}_{ii}^\text{BS}-\hat{\mathbf{G}}_{ii}^\text{BS}) (\mathbf{\Phi}_i^\text{SI})^{T} \nonumber \\  
&&\hspace{3.5cm}+\sqrt{\tau_\text{max}P_d} \mathbf{G}_{ii}^\text{BS}(\mathbf{\Psi}_i^\text{SPT})^{T} + \mathbf{N}_i + \mathbf{\Delta}_i^\text{SPT}. 
\end{eqnarray} 
\begin{theorem} \label{theorem: 1}
By means of  SPT, the distribution of the estimated  $m^\text{th}$ and  $k^\text{th}$ columns of $\mathbf{G}^\text{BS}_{i,i}$ and $\mathbf{G}^u_{i,i}$ are given by 
\begin{subequations}
\begin{eqnarray}
&&\hspace{-.7cm}\hat{{g}}_{ii,\ell m}^\text{BS} \sim \mathcal{CN}\left(0,\frac{\tau_\text{max}{P_d} (\rho_{ii,\ell m}^\text{BS})^4}{P_u \sum_{j=1}^N \sum_{k=1}^{K_\text{UL}} (\rho_{ij,k}^{u})^2 + \text{A} } \right), \label{eq: est_SPT_a}\\
&&\hspace{-.7cm}\hat{\mathbf{g}}_{ii,k}^u \sim 
\mathcal{CN}\left(0, \frac{\tau_\text{max}P_u (\rho_{ii,k}^{u})^4}{\tau_\text{max}P_u \sum_{j=1}^N (\rho_{ij,k}^{u})^2 + P_d \sum_{m=1}^{M_t} (\rho_{ii,m}^\text{BS})^2\mathrm{Var}[\tilde{\epsilon}_{ii,m}^\text{BS}] + \text{A}} \mathbf{I}_{M_r}\right),  
\end{eqnarray}
\end{subequations}
where $\text{A}=(1+\beta)\{\tau_\text{max} P_d (\rho_{ii,lm}^\text{BS})^2 +  \alpha P_d\sum_{m=1}^{M_t}(\rho_{ii,m}^\text{BS})^2  + n_o\}$. 
 We omit the distribution of the estimated DL channel because it is the same as those of  the nSPT case. 
\end{theorem}

\begin{proof}
See Appendix A. \end{proof}
\begin{remark}[SI channel estimation error of nSPT and SPT]\label{rem: NMSE}
By means of the deriving normalized minimum mean square error (NMSE) which is defined as $|g-\hat{g}|^2 / |g|^2$, we can measure the level of SI cancellation for nSPT and SPT.  The NMSE values of nSPT and SPT from (\ref{eq: est_nSPT_c})  and (\ref{eq: est_SPT_a}), respectively, are given by 

\begin{subequations}
\begin{eqnarray}
&&\hspace{-1.2cm}\xi_{nSPT}^2=\frac{ \beta \tau_{SI} P_d (\rho_{ii,\ell m}^{BS})^2 + (1+\beta)\{\alpha P_d \sum_{m=1}^{M_t}(\rho_{ii,\ell m}^{BS})^2 + n_0 \}   } { (1+ \beta)\{\tau_{SI} P_d (\rho_{ii,\ell m}^{BS})^2 + \alpha P_d \sum_{m=1}^{M_t}(\rho_{ii,\ell m}^{BS})^2 + n_0 \} }, \label{eq: NMSE_nSPT} \\
&&\hspace{-1.2cm}\xi_{SPT}^2 = \frac{P_u \sum_{j=1}^N \sum_{k=1}^{K_\text{UL}} (\rho_{ij,k}^{u})^2 + \beta \tau_\text{max} P_d (\rho_{ii,\ell m}^\text{BS})^2 +     (1+\beta)\{\alpha P_d \sum_{m=1}^{M_t}(\rho_{ii,\ell m}^\text{BS})^2 + n_o\}}{ P_u \sum_{j=1}^N \sum_{k=1}^{K_\text{UL}} (\rho_{ij,k}^{u})^2 + (1+\beta)\{\tau_\text{max} P_d (\rho_{ii,\ell m}^\text{BS})^2 +\alpha P_d \sum_{m=1}^{M_t}(\rho_{ii,\ell m}^\text{BS})^2 + n_o\} }. \label{eq: NMSE_SPT}
\end{eqnarray}
\end{subequations}
Contrary to nSPT, there exists an additional interference term in (\ref{eq: NMSE_SPT}) induced by UL pilot transmission in SPT. 
However, this term can be negligible for the following reasons: i) this term is unrelated to pilot overhead $\tau_{\text{max}}$ which means that it will become relatively small as $M$ increases. ii) the large-scale fading gain of UL users, $(\rho_{ij,k}^u)^2$, is relatively smaller than large-scale fading gain of SI channel, $(\rho_{ii,lm}^\text{BS})^2$. Especially, for cell-boundary users, large-scale fading gain become much smaller. In conclusion, we can as easily obtain the performance of SI cancellation by using SPT as by using nSPT. 
\end{remark}

\section{Analytic model for ergodic achievable sum-rate} \label{sec: sum-rate}
By applying simple MF and ZF linear filters to BS for two multicell system scenarios, we introduce the analytic model for ergodic achievable sum-rate of cell-boundray users.

\subsection{System Scenario 1: Non-Cooperative Multicell System}  \label{subsec: NoCRAN}
\subsubsection{Ergodic Achievable Downlink Sum-Rate}
Based on \cite{lim2015performance}, we adopt matrix-normalization for the MF precoder and vector-normalization for the ZF precoder. In other words, for cell $i$,  $\mathbf{F}_i^\text{MF}=(\hat{\mathbf{G}}_{ii}^d)^{H}/||\hat{\mathbf{G}}_{ii}^d||$ and $\mathbf{F}_i^\text{ZF} = (\hat{\mathbf{G}}_{ii}^d)^{H}(\hat{\mathbf{G}}_{ii}^d (\hat{\mathbf{G}}_{ii}^d)^{H})^{-1}/||\mathbf{F}_v|| = \Big[\frac{\mathbf{f}_{i,1}}{\sqrt{K_\text{DL}}||\mathbf{f}_{i,1}||} \ldots \frac{\mathbf{f}_{i,K_\text{DL}}}{\sqrt{K_\text{DL}}||\mathbf{f}_{i,K_\text{DL}}||} \Big]$.
\begin{theorem} \label{theorem: 2}
From (\ref{eq:y_NoCRAN}), the DL ergodic achievable sum-rate for $K_\text{DL}$ cell-boundary users in cell $i$ is given as follows: for MF precoder, 
\begin{subequations}
\begin{eqnarray} \label{eq:R_d_MF}
&&\hspace{-2cm}R_{d,i}^\text{MF} \approx  \sum_{k=1}^{K_\text{DL}}  \mathcal{C} \Bigg( \frac{P_d (\hat{\rho}_{ii,k}^d)^4 M_t(M_t+1)}{I_{d,i}^\text{MF} + M_t P_d (\hat{\rho}_{ii,k}^d)^2 (\bar{\rho}_{ii,k}^d)^2 + M_t \sum_{k=1}^{K_\text{DL}} (\hat{\rho}_{ii,k}^d)^2 {n_0} }  \Bigg) \\
&&\hspace{-2cm}I_{d,i}^\text{MF} = P_d M_t (\rho_{ii,k}^d)^2\sum_{l=1,l\neq k}^{K_\text{DL}}(\hat{\rho}_{ii,l}^d)^2 + P_d M_t \sum_{k=1}^{K_\text{DL}} (\hat{\rho}_{ii,k}^d)^2 \sum_{j=1,j\neq i}^N (\rho_{ij,k}^d)^2  \nonumber\\ 
&& \hspace{6cm}+ P_u M_t \sum_{k=1}^{K_\text{DL}} (\hat{\rho}_{ii,k}^d)^2 \sum_{j=1}^N \sum_{k=1}^{K_\text{UL}} (\rho_{ij,k}^\text{UE})^2, 
\end{eqnarray}
\end{subequations}
where $\hat{\mathbf{g}}_{ii,k}^d \sim \mathcal{CN}\big(0, (\hat{\rho}_{ii,k}^d)^2 \mathbf{I}_{M_r}\big)$ and ${\mathbf{\epsilon}}_{ii,k}^d \sim \mathcal{CN}\big(0, (\bar{\rho}_{ii,k}^d)^2 \mathbf{I}_{M_r}\big)$. For ZF precoder, 
\begin{eqnarray} \label{eq:R_d_ZF}
R_{d,i}^\text{ZF} \approx \sum_{k=1}^{K_\text{DL}} \mathcal{C} \Bigg( \frac{P_d (\hat{\rho}_{ii,k}^d)^2 \frac{M_t-K_\text{DL}+1}{K_\text{DL}}}{P_d \sum_{j=1, j\neq i}^N (\rho_{ij,k}^d)^2 + P_u \sum_{j=1}^N \sum_{k=1}^{K_\text{UL}} (\rho_{ij,k}^\text{UE})^2 + P_d (\bar{\rho}_{ii,k}^d)^2 + {n_0}} \Bigg ).
\end{eqnarray}
\end{theorem}
\begin{proof}
See Appendix B. \end{proof}
\subsubsection{Ergodic Achievable Uplink Sum-Rate}
We define the MF and ZF detection filter as $\mathbf{W}_{i}^\text{MF}=(\hat{\mathbf{G}}_{i,i}^u)^{H}$ and $\mathbf{W}_i^\text{ZF}=(\hat{\mathbf{G}}_{ii}^u)^{H}(\hat{\mathbf{G}}_{ii}^u(\hat{\mathbf{G}}_{ii}^u)^{H})^{-1}$. 
\begin{theorem} \label{theorem: 3}
From (\ref{eq:y_NoCRAN_ul}), the UL ergodic achievable sum-rate for $K_\text{UL}$ cell-boundary users in cell $i$ is given as follows: For MF detection filter, 
\begin{subequations}
\begin{eqnarray} \label{eq:R_u_MF}
&&\hspace{-1cm}R_{u,i}^\text{MF} \approx \sum_{k=1}^{K_\text{UL}} \mathcal{C} \Bigg( \frac{P_u M_r (\hat{\rho}_{ii,k}^u)^2}{P_u \sum_{\ell=1,\ell \neq k}^{K_\text{UL}} (\rho_{ii,\ell}^u)^2 + P_u \sum_{j=1, j\neq i}^N \sum_{k=1}^{K_\text{UL}} (\rho_{ij,k}^u)^2+ I_{u,i}^\text{MF} + P_u (\bar{\rho}_{ii,k}^u)^2 + {n_0} } \Bigg), \\
&&\hspace{-1cm}I_{u,i}^\text{MF} = \frac{P_d}{M_t} \sum_{j=1, j\neq i}^N \sum_{\ell=1}^{M_r}  (\rho_{ij,m}^\text{BS})^2 + \frac{P_d}{M_t M_r} \sum_{\ell=1}^{M_r} \sum_{m=1}^{M_t} (\bar{\rho}_{ii,\ell m}^\text{BS})^2+  \alpha \frac{P_d}{M_t M_r}\sum_{\ell=1}^{M_r}\sum_{m=1}^{M_t} (\rho_{ii,\ell m}^\text{BS})^2 +  \nonumber \\ 
&&\hspace{6.5cm}\beta \left\{ (1+\alpha)\frac{P_d}{M_t M_r}\sum_{\ell=1}^{M_r}\sum_{m=1}^{M_t} (\rho_{ii,\ell m}^\text{BS})^2  +{n_0}\right\}, 
\end{eqnarray} 
\end{subequations}
where $\hat{\mathbf{g}}_{ii,k}^u \sim \mathcal{CN}\big(0, (\hat{\rho}_{ii,k}^u)^2 \mathbf{I}_{M_r}\big)$ and ${\mathbf{\epsilon}}_{ii,k}^u \sim \mathcal{CN}\big(0, (\bar{\rho}_{ii,k}^u)^2 \mathbf{I}_{M_r}\big)$. For ZF detection filter, 
\begin{eqnarray} \label{eq:R_u_ZF}
&&\hspace{-0.7cm}R_{u,i}^\text{ZF} \approx 
\sum_{k=1}^{K_\text{UL}} \mathcal{C} \Bigg(\frac{P_u}{\frac{1}{(\hat{\rho}_{ii,k}^u)^2(M_r-K_{UL}+1)} I_{u,i}^\text{ZF} + \frac{P_u}{(\hat{\rho}_{ii,k}^u)^2(M_r-K_{UL}+1)}\sum_{\ell=1}^{K_{UL}}(\bar{\rho}_{ii,\ell}^u)^2 + \frac{{n_0}}{(\hat{\rho}_{ii,k}^u)^2(M_r-K_{UL}+1)}} \Bigg), \\
&&\hspace{-0.7cm}I_{u,i}^\text{ZF} =   P_u \sum_{j=1, j\neq i}^N\sum_{k=1}^{K_\text{UL}}(\rho_{ij,k}^u)^2   + \frac{P_d}{M_t} \sum_{j=1, j\neq i}^N \sum_{m=1}^{M_t}  (\rho_{ij,m}^\text{BS})^2 + \frac{P_d}{M_t M_r} \sum_{\ell=1}^{M_r}\sum_{m=1}^{M_t}  (\bar{\rho}_{ii,lm}^\text{BS})^2+ \nonumber\\ 
&&\hspace{3.5cm}\alpha \frac{P_d}{M_t M_r} \sum_{\ell=1}^{M_r} \sum_{m=1}^{M_t} (\rho_{ii,\ell m}^\text{BS})^2 + \beta\left\{ (1+\alpha) \frac{P_d}{M_t M_r} \sum_{\ell=1}^{M_r}\sum_{m=1}^{M_t} (\rho_{ii,\ell m}^\text{BS})^2 +  {n_0}\right\}, \nonumber
\end{eqnarray}
where ${\epsilon}_{ii,\ell m}^\text{BS} \sim \mathcal{CN} (0, (\bar{\rho}_{ii,\ell m}^\text{BS})^2 )$.
\end{theorem} 
\begin{proof}
See Appendix C. \end{proof}
\subsection{System Scenario 2: Cooperative Multicell System }
Considering a full-centralized CRAN system, a precoder and detection filter are produced based on the global CSI of the system at CU. Unlike the case of the non-cooperative multicell system, there exists  DL and UL quantization noise occurring at front-haul due to limited front-haul capacity. 
\subsubsection{Ergodic Achievable Downlink Sum-Rate}
We define the MF precoder as  $\mathbf{F}^\text{MF} = (\hat{\mathbf{G}}^d)^H/||\hat{\mathbf{G}}^d||$ and the ZF precoder as $\mathbf{F}^\text{ZF}=(\hat{\mathbf{G}}^d)^{H}(\hat{\mathbf{G}}^d(\hat{\mathbf{G}}^d)^{H})^{-1}/||\mathbf{F}_v||$.  
\begin{theorem} \label{theorem: 4}
From (\ref{eq:y_CRAN}), the DL ergodic achievable sum-rate for $K_\text{DL}$ cell-boundary users in cell $i$ is given as follows: For MF precoder, 
\begin{subequations}
\begin{eqnarray} \label{eq:R_bar_d_MF}
&&\hspace{-1cm}{R}_{d,i}^\text{C,MF}  \approx \\
&&\sum_{k=1}^{K_\text{DL}} \mathcal{C} \left( \frac{P_s\{\sum_{j=1}^N M_t(M_t+1)(\hat{\rho}_{ij,k}^d)^4 + \sum_{(n,m)\in \Omega_1} M_t^2(\hat{\rho}_{in,k}^d)^2(\hat{\rho}_{im,k}^d)^2\}}{\mathrm{I_{d,i}^\text{C,MF}} + {I_{q,d,i}^\text{C,MF}} + P_s \sum_{j=1}^N M_t (\bar{\rho}_{ij,k}^d)^2(\hat{\rho}_{ij,k}^d)^2 + \sum_{(i,j)=(1,1)}^{(N,N)} \sum_{k=1}^{K_\text{DL}} M_t(\hat{\rho}_{ij,k}^d)^2 {n_0}  } \right) \nonumber\\
&&\hspace{-1cm}{I_{d,i}^\text{C,MF}}=P_s \sum_{(n,j)\in \Omega_2} \sum_{M=1}^N M_t(\rho_{im,k}^d)^2(\hat{\rho}_{nm,j}^d)^2+ P_u \left\{\sum_{j=1}^N \sum_{k=1}^{K_\text{UL}}(\rho_{ij,k}^\text{UE})^2\right\}\sum_{(i,j)=(1,1)}^{(N,N)} \sum_{k=1}^{K_\text{DL}} M_t(\hat{\rho}_{ij,k}^d)^2 \label{eq:I_d_FD}  \\ 
&&\hspace{-1cm}{I_{q,d,i}^\text{C,MF}}= 
\sum_{j=1}^N \left(M_t (\sigma_j^{d,\text{MF}})^2 (\rho_{ij,k}^d)^2 \right), \\ 
&&\hspace{-1cm}(\sigma_j^{d,\text{MF}})^2=\frac{P_s \sum_{n=1}^N \big \{ \sum_{k=1}^{K_\text{DL}}M_t(\rho_{in,k}^d)^2 \big \} /\sum_{(i,j)=(1,1)}^{(N,N)} \sum_{k=1}^{K_\text{DL}} M_t(\hat{\rho}_{ij,k}^d)^2}{2^{C_d}-1},  \label{eq:I_q}
\end{eqnarray}
\end{subequations}
where, $\Omega_1= \big \{ (n,m)|n=1 \ldots N, m=1 \ldots N, n \neq m\big\}$ and $\Omega_2= \big \{ (n,j)|n=1 \ldots N, j=1 \ldots K_\text{DL}, (n,j) \neq (i,k) \big \}$. For ZF precoder, 
\begin{eqnarray} \label{eq:R_bar_d_ZF}
{R}_{d,i}^\text{C,ZF} \approx \sum_{k=1}^{K_\text{DL}} \mathcal{C} \left( \frac{P_s (\hat{\rho}_{ik,\text{avg}}^d)^2 (M_tN-K_\text{DL}N+1)/(K_\text{DL}N)}{P_u \sum_{j=1}^N \sum_{k=1}^{K_\text{UL}}(\rho_{ij,k}^\text{UE})^2 + \sum_{j=1}^N \big\{M_t (\sigma_j^{d,\text{ZF}})^2 (\rho_{ij,k}^d)^2 \big\} + P_s (\bar{\rho}_{ik,\text{avg}}^d)^2+{n_0}} \right),
\end{eqnarray}
where $(\sigma_j^{d,\text{ZF}})^2 = P_s / N, (\hat{\rho}_{ik,\text{avg}}^d)^2= \mathbb{E} \big[ (\hat{\rho}_{in,k}^d)^2\big]$ and $(\bar{\rho}_{ik,\text{avg}}^d)^2= \mathbb{E}\big[(\bar{\rho}_{in,k}^d)^2\big]$.
\end{theorem}
\begin{proof}
See Appendix D. \end{proof}

\subsubsection{Ergodic Achievable Uplink Sum-Rate}
We define the MF and ZF detection filter as $\mathbf{W}^\text{MF}=(\hat{\mathbf{G}}^u)^H$ and $\mathbf{W}^\text{ZF}= (\hat{\mathbf{G}}^u)^H (\hat{\mathbf{G}}^u (\hat{\mathbf{G}}^u)^H)^{-1}$. 
\begin{theorem} \label{theorem: 5}
From (\ref{eq:y_CRAN_ul}), the UL ergodic achievable sum-rate of $K_\text{UL}$ cell-boundary users for cell $i$ is given as follows: For MF detection, 
\begin{subequations}
\begin{eqnarray} \label{eq:R_bar_u_MF}
&&\hspace{-0.9cm}{R}_{u,i}^\text{C,MF} \approx \\ 
&&\hspace{-0.5cm}\sum_{k=1}^{K_\text{UL}} \mathcal{C} \left( \frac{P_u \sum_{n=1}^N M_r(\hat{\rho}_{in,k}^u)^2}{P_u \frac{\sum_{(n,j)=(1,1)}^{(N,K_\text{UL})} \left \{ \sum_{m=1}^N (\hat{\rho}_{im,k}^u)^2 (\rho_{nm,j}^u)^2 M_r \right \} }{\sum_{n=1}^N (\hat{\rho}_{in,k}^u)^2 M_r} + {I_{u,i}^\text{C,MF}} + {I_{q,u,i}^\text{C,MF}} + P_u\frac{\sum_{m=1}^N (\hat{\rho}_{im,k}^u)^2 (\bar{\rho}_{im,k}^u)^2}{ \sum_{n=1}^N (\hat{\rho}_{in,k}^u)^2} +{n_0}} \right) \nonumber
\end{eqnarray}
\begin{eqnarray}
&&\hspace{-.6cm}{I_{u,i}^\text{C,MF}}=
\frac{P_s}{A}\sum_{j=1}^N \left\{\left( M_t M_r\sum_{n=1, n\neq j}^N (\hat{\rho}_{in,k}^u)^2 (\rho_{jn}^\text{BS})^2 + \sum_{\ell=1}^{M_r} \sum_{m=1}^{M_t}(\hat{\rho}_{ij,k}^u)^2 (\bar{\rho}_{jj,\ell m}^\text{BS})^2 \right ) 
\left( \sum_{n=1}^N\sum_{k=1}^{K_\text{DL}}(\hat{\rho}_{nj,k}^d)^2 \right ) \right \}   \nonumber  \\
&&\hspace{2cm}+(\beta + \alpha \beta) \frac{P_s}{A}  \sum_{j=1}^N\left\{\left(\sum_{\ell=1}^{M_r} \sum_{m=1}^{M_t}(\hat{\rho}_{ij,k}^u)^2 ({\rho}_{jj,\ell m}^\text{BS})^2\right) \left(\sum_{n=1}^N \sum_{k=1}^{K_\text{DL}}(\hat{\rho}_{nj,k}^d)^2  \right) \right\}  + \beta,   
\end{eqnarray}
\begin{eqnarray}
&&\hspace{-.3cm}{I_{q,u,i}^\text{C,MF}}=\sum_{m=1}^{N} (\sigma_{m}^{u,\text{MF}})^2 (\hat{\rho}_{im,k}^u)^2 / \sum_{n=1}^N (\hat{\rho}_{in,k}^u)^2,  
\end{eqnarray}
\begin{eqnarray}
&&\hspace{-.7cm}(\sigma_m^{u,\text{MF}})^2= \frac{1}{2^{C_u}-1} \Bigg\{ \frac{P_s }{\text{E}||\hat{\mathbf{H}}^d||^2} \Big\{ M_tM_r\sum_{n=1, n\neq i}^N \big ( \sum_{j=1}^N\sum_{k=1}^{K_\text{DL}}(\hat{\rho}_{jn,k}^d)^2\big )  (\rho_{in}^\text{BS})^2  +(\hat{\rho}_{ji,k}^d)^2 \sum_{\ell=1}^{M_r} \sum_{m=1}^{M_t} (\rho_{ii,\ell m}^\text{BS})^2  \nonumber\\
&& \hspace{1.3cm} +(\beta+\alpha \beta) (\rho_{ii}^\text{BS})^2 \sum_{j=1}^N \sum_{k=1}^{K_\text{DL}} (\hat{\rho}_{ji,k}^d)^2    \Big\}  
+P_u \big \{ \sum_{j=1}^N \sum_{k=1}^{K_\text{UL}}(\rho_{ij,k}^u)^2 \big \} M_r +(1+\beta) M_r{n_0} \Bigg\}\label{eq:q_d_CRAN} \label{eq: sigma_ul}, 
\end{eqnarray}
\end{subequations}
where $A=(\sum_{n=1}^N (\hat{\rho}_{in,k}^u)^2 M_r)\text{E}||\hat{\mathbf{H}}^d||^2$, and $\text{E}||\hat{\mathbf{H}}^d||^2= \sum_{n=1}^N \sum_{j=1}^N \sum_{k=1}^{K_\text{DL}} (\hat{\rho}_{jn,k}^d)^2 M_t$. 
For ZF detection, 
\begin{eqnarray} \label{eq:R_bar_u_ZF}
{R}_{u,i}^\text{C,ZF} \approx\hspace{-0.15cm} \sum_{k=1}^{K_\text{UL}} \mathcal{C} \left( \frac{P_u (\hat{\rho}_{ik,\text{avg}}^u)^2(M_rN-K_{UL}N+1)}{P_s\frac{g_\text{avg}}{M_tN} + P_u \sum_{(n,j)=(1,1)}^{N,K_\text{UL}}(\bar{\rho}_{nj,\text{avg}}^u)^2 + (\sigma_\text{avg}^{u, ZF})^2 + ( \beta + \alpha \beta) \frac{g_\text{avg}^\text{diag}}{M_tN} + (\beta +1){n_0} } \right )\hspace{-0.15cm},
\end{eqnarray}\\
where $ (\bar{\rho}_{ik,\text{avg}}^u)^2=\mathbb{E}\big[(\bar{\rho}_{in,k}^u)^2 \big]$,   $(\hat{\rho}_{ik,\text{avg}}^u)^2=\mathbb{E} \big[(\hat{\rho}_{in,k}^u)^2\big]$. Also, we define $g_{avg}=\mathbb{E}\big[g_n \big]$, where $g_i = M_t\sum_{j=1, j\neq i}^N  (\rho_{ji}^\text{BS})^2 + \sum_{m=1}^{M_t} (\bar{\rho}_{ii,lm}^\text{BS})^2$. In (\ref{eq:R_bar_u_ZF}),  we define $g_\text{avg}^\text{diag} = \mathbb{E}[g_1^\text{diag} \ldots g_N^\text{diag}]$, where $g_i^\text{diag}=\sum_{\ell=1}^{M_r}\sum_{m=1}^{M_t} (\rho_{ii,\ell m}^\text{BS})^2$. The average quantization noise is defined as $(\sigma_\text{qvg}^{u,ZF})^2 = \mathbb{E}\big[(\sigma_n^{u,ZF})^2\big]$, where $q_i^u = \frac{\mathbb{E}||\mathbf{y}_i^u||^2}{2^{C_u}-1}$  based on (\ref{eq: y_ul_i}),
\begin{eqnarray}
&&\mathbb{E}||\mathbf{y}_i^u||^2= P_u M_r \sum_{j=1}^N \sum_{k=1}^{K_\text{UL}} (\rho_{ij,k}^u)^2  + \frac{P_s}{M_t N} \left(M_t M_r \sum_{j=1}^N (\rho_{ij}^\text{BS})^2 + \sum_{\ell=1}^{M_r}\sum_{m=1}^{M_t} (\rho_{ii,\ell m}^\text{BS})^2 \right) +  \\&&\hspace{5cm}\frac{\alpha P_s }{M_t N} \sum_{\ell=1}^{M_r}\sum_{m=1}^{M_t} (\rho_{ii,\ell m}^\text{BS})^2 +
\beta(1+\alpha) \frac{M_r}{M_t N} g_\text{avg}^\text{diag} + (1+\beta) M_r {n_0}. \nonumber
\end{eqnarray}
\end{theorem}
\begin{proof}
See Appendix E. \end{proof}
\section{Performance Analysis} \label{sec: analysis}
The performance of the full-duplex cellular system heavily depends  on the channel estimation method and the level of the aforementioned interference-- SI and BS-BS interference and user-user interference. Here, we discuss the feasibility of a  full-duplex large-scale MIMO system by analyzing the ergodic achievable sum-rate based on obtained analytic results. Furthermore, we investigate the  reliable region which indicates the interval that can gaurantee a better sum-rate from a full-duplex system than from a half-duplex system.  For comparison, without loss of generality, we assume that the same linear filter is applied to both system. Also, we assume that all cells operate either as half-duplex or full-duplex systems  \cite{Lee2015hybrid}. Basically, the reliable region when considering pilot overhead can be obtained by solving the following inequality.
\begin{eqnarray} \label{eq:condition}
2(1-\frac{\tau_{FD}}{T}) \left(\sum_{k=1}^{K_\text{DL}}\mathrm{R_{d,k}^\text{FD}}+\sum_{k=1}^{K_\text{UL}}\mathrm{R_{u,k}^\text{FD}} \right) \ge (1-\frac{\tau_\text{HD}}{T}) \left(\sum_{k=1}^{K_\text{DL}}\mathrm{R_{d,k}^\text{HD}}+\sum_{k=1}^{K_\text{UL}}\mathrm{R_{u,k}^\text{HD}} \right), 
\end{eqnarray}
where $\tau_\text{FD}, \tau_\text{HD}, R_{d,k}^\text{FD}$ and $R_{u,k}^\text{FD}$ denote the pilot overhead of the full-duplex and the half-duplex system, and the achievable rate of the full-duplex system for the DL and UL user $k$, respectively.  $T$ denotes the number of total symbols per coherence time, $T_\text{cohe}$, which includes all pilots and data symbols.  Since we assume that symbol duration is unchanged, the increase of $T$ implies the increase of $T_\text{cohe}$. In (\ref{eq:condition}), we define $R_{d,k}^\text{FD}=\mathcal{C} \Big( \text{SINR}_{\text{DL},k}^\text{FD} \Big)$, where $\text{SINR}_{\text{DL},k}^\text{FD}={S_{\text{DL},k}^\text{FD}}/(I_{\text{DL},k} + I_{\text{DL},k}^\text{FD})$, and $R_{u,k}^\text{FD}=\mathcal{C} \Big( \text{SINR}_{\text{UL},k}^\text{FD} \Big)$, where $\text{SINR}_{\text{UL},k}^\text{FD}=S_{\text{UL},k}^\text{FD}/(I_{\text{UL},k} + I_{\text{UL},k}^\text{FD})$. $S_{\text{DL},k}^\text{FD}$ and $I_{\text{DL},k}$ denote the power of the desired signal of user $k$ and the sum of power of all pre-existing DL interference such as intra and intercell interference, channel estimation error, and noise.  $I_{\text{DL},k}^\text{FD}$ denotes the sum of the power of DL interference induced by full-duplex BS.  In this regard, $I_{\text{UL},k}^\text{FD}$ denotes the sum of  the power of SI, BS-BS interference, Tx noise, and Rx distortion. In the half-duplex cellular system, we denote the SINR of user $k$ as $\text{SINR}_{\text{DL},k}^\text{HD}=S_{\text{DL},k}^\text{HD}/I_{\text{DL},k}^\text{HD}$ and $\text{SINR}_{\text{UL},k}^\text{HD}=S_{\text{UL},k}^\text{HD}/I_{\text{UL},k}^\text{HD}$. 
In order to determine the reliable region, we partition (\ref{eq:condition}) into DL and UL transmission. Otherwise, even if we obtain a more achievable sum-rate for a full-duplex system,  the achievable sum-rate of either DL or UL could be smaller than those of the half-duplex system. Thus, we derive conditions for the reliable region that satisfy both 
\begin{eqnarray}
2(1-\frac{\tau_\text{FD}}{T}) \sum_{k=1}^{K_\text{DL}}{R_{d,k}^\text{FD}} \ge (1-\frac{\tau_\text{HD}}{T}) \sum_{k=1}^{K_\text{DL}}{R_{d,k}^\text{HD}} \quad \text{and} \label{eq: condition_dl}\\ 
2(1-\frac{\tau_\text{FD}}{T}) \sum_{k=1}^{K_\text{UL}}{R_{u,k}^\text{FD}}  \ge (1-\frac{\tau_\text{HD}}{T}) \sum_{k=1}^{K_\text{UL}}{R_{u,k}^\text{HD}}. \nonumber \label{eq: condition_ul}
\end{eqnarray}

\begin{lemma}[nSPT] 
Using the nSPT, the maximum tolerant power  of the interference and the minimum  required coherence time for the both multicell scenarios is 
\begin{subequations}
\begin{eqnarray}
&&\hspace{-1cm}{\text{I}_{\text{DL},k}^\text{FD}} \leq \big(\frac{T^\text{nSPT}-2(\tau_\text{SI}+\tau_\text{UU})-\tau_\text{UD} }{T^\text{nSPT}-\tau_\text{UD} }\big) {\text{I}_{\text{DL},k}}, \hspace{0.2cm}{\text{I}_{\text{UL},k}^\text{FD}} \leq \big(\frac{T^\text{nSPT}-2(\tau_\text{SI}+\tau_\text{UD})-\tau_\text{UU} }{T^\text{nSPT}-\tau_\text{UU} }\big) {\text{I}_{\text{UL},k}}, \label{eq: lem1_a}\\
&&\hspace{-1cm}T_\text{cohe} > T_{\text{cohe}}^{\text{nSPT}}=t_s \max \{ 2(\tau_\text{SI} + \tau_\text{UU})+ \tau_\text{UD},\; 2(\tau_\text{SI} + \tau_\text{UD})+ \tau_\text{UU} \} , \label{eq: lem1_b}
\end{eqnarray}
\end{subequations}
where  $T^\text{nSPT}$ and $t_s$ denote the total number of symbols for the nSPT and  symbol durations [sec/symbol].
\end{lemma}
\begin{proof}
Based on  Bernoulli's inequality, where $(1+x)^r \ge 1+rx$ for every integer $r$ and every real number $x \ge -1$, we reformulate (\ref{eq: condition_dl}) as  
\begin{eqnarray}\label{eq: proof}
2\left(1-\frac{\tau_\text{FD}}{T}\right) {\text{SINR}_{\text{x},k}^\text{FD}} \geq \left(1-\frac{\tau_\text{HD}}{T}\right) {\text{SINR}_{\text{x},k}^\text{HD}}
\end{eqnarray}
for cell-boundary DL user $k$, where $\text{x}\in\{\text{DL}, \text{UL}\}$.  Since, with nSPT, we send the pilot orthogonally in the time domain, the channel estimation errors of DL and UL in the half-duplex system are the same as those of the full-duplex system with nSPT. We then say that  
(${S_{\text{DL},k}^\text{FD}} = {S_{\text{DL},k}^\text{HD}}$,  ${I_{\text{DL},k}} = \mathrm{I_{\text{DL},k}^\text{HD}}$) and (${S_{\text{UL},k}^\text{FD}} = {S_{\text{UL},k}^\text{HD}}$,  ${I_{\text{UL},k}} = \mathrm{I_{\text{UL},k}^\text{HD}}$) for all user $k$. After we reorganize (\ref{eq: proof}) in terms of $I^\text{FD}_\text{x,k}$,  we obtain (\ref{eq: lem1_a}). Then, since the coefficient of $I_\text{x,k}$ needs to be larger than zero in order to satisfy (\ref{eq: lem1_a}), we obtain (\ref{eq: lem1_b}) from (\ref{eq: lem1_a}). 
\end{proof}

\begin{lemma}[SPT]
Using the SPT, the maximum tolerant power  of the interference and the minimum required coherence time for  both multicell scenarios is 
\begin{subequations}
\begin{eqnarray}
&&\hspace{-0.5cm}{\text{I}_{\text{DL},k}^\text{FD}} \leq \left(\frac{T^{SPT}-2\tau_\text{max}-\tau_\text{UD} }{T^{SPT}-\tau_\text{UD} }\right) {\text{I}_{\text{DL},k}} \;\; \mathrm{and} \;\;
{I_{\text{UL},k}^\text{FD}} \leq \frac{2(T^{SPT}-\tau_\text{max}-\tau_\text{UD}) \hspace{0.1cm}{S_{\mathrm{UL},k}^\mathrm{FD}} }{ (T^{SPT}-\tau_\text{UU})\hspace{0.1cm}{\mathrm{SINR_{UL,k}^{HD}} } }-{\text{I}_{\text{UL},k}}, \label{eq: lem2_a}\\
&&\hspace{-0.5cm} T_\text{cohe} \geq T_\text{cohe}^{\text{SPT}}=t_s \max \{ 2\tau_{\max}+\tau_\text{UD},\; T_\text{max}^\text{UL}\} , \label{eq: lem2_b}\\ 
&&\hspace{-0.4cm}\textrm{where} \;\;T_{\max}^\text{UL} = \max \left \{ \frac{2\tau_\text{FD} \mathrm{SINR_{UL,1}^{FD}} - \tau_\text{HD} \mathrm{SINR_{UL,1}^{HD}} }{ 2\mathrm{SINR_{UL,1}^{FD}} - \mathrm{SINR_{UL,1}^{HD}} }, \ldots , \frac{2\tau_\text{FD} \mathrm{SINR}_{\mathrm{UL},K}^\mathrm{FD} - \tau_\mathrm{HD} \mathrm{SINR}_{\mathrm{DL},K}^\mathrm{HD} }{ 2\mathrm{SINR}_{\mathrm{DL},K}^\mathrm{FD} - \mathrm{SINR}_{\mathrm{DL},K}^\mathrm{HD} }  \right\}. \nonumber
\end{eqnarray}
\end{subequations}
$T^\text{nSPT}$ denotes the total number of symbol for the nSPT. The analytic value of  $\mathrm{SINR_{UL,k}^{FD}}$ and $\mathrm{SINR_{UL,k}^{HD}}$ are given in Sec. \ref{sec: sum-rate}. We recall that $\tau_{\max}=\max(\tau_\text{SI}, \tau_\text{UU})$. 
\end{lemma}
\begin{proof}
In a  manner similar to that of Lemma 1, we start with the proof from (\ref{eq: proof}). To estimate the DL channels, we have (${S_{\text{DL},k}^\text{FD}} = {S_{\text{DL},k}^\text{HD}}$,  ${I_{\text{DL},k}} = \mathrm{I_{\text{DL},k}^\text{HD}}$) for the same reason as in  Lemma 1. However, to estimate the UL channels, we cannot calculate in the same way due to a different resulting channel estimation error. Thus, the second inequality of (\ref{eq: lem2_a}) is obtained by reorganizing (\ref{eq: proof}) in terms of $\mathrm{I_{UL,k}^{FD}}$. $\mathrm{T_{\max}^{UL}}$ in (\ref{eq: lem2_b}) is obtained by reorganizing (\ref{eq: proof}) in terms of $T$. 
\end{proof}

\begin{remark}[Sensitivity to the Coherence Time]
Comparing (\ref{eq: lem1_b}) with (\ref{eq: lem2_b}), we first observe that $2\tau_{\max}+\tau_\text{UD} < 2(\tau_\text{SI}+\tau_\text{UU})+\tau_\text{UD}$. Next, from (\ref{eq: lem2_b}),  $T_\mathrm{max}^\mathrm{UL}= (2\tau_\text{FD}-\alpha\tau_\text{HD} )(2-\alpha)$, where $\alpha=\mathrm{SINR_{UL,k}^{HD}}/\mathrm{SINR_{UL,k}^{FD}} \le 1$ and $k$ is the maximization index. We  obtain the maximum value of $\mathrm{T_{max}^{UL}}$ as $2\tau_\text{FD}-\tau_\text{HD}=2\tau_\text{max}-\tau_\text{UD}$ when $\alpha=1$, which is smaller than $2(\tau_\text{SI}+\tau_\text{UD})+\tau_\text{UU}$ in (\ref{eq: lem1_a}). Accordingly, we say that the SPT is less sensitive to the coherence time than  to nSPT. i.e. $T_\text{cohe}^\text{SPT} \le T_\text{cohe}^\text{nSPT}$. 
\end{remark}

\begin{lemma}[Asymptotic Achievable Sum-Rate of Full-Duplex System in Non-Cooperative Multicell] \label{lem: M1}
From (\ref{eq:R_d_MF}) and (\ref{eq:R_d_ZF}), we obtain the asymptotic result on the downlink ergodic achievable sum-rate given as 
\begin{eqnarray}
&&\lim_{M\to \infty} R_{d,i}^\text{MF} \to \sum_{k=1}^{K_\text{DL}} \log_2\left(1+ \frac{P_r (\hat{\rho}_{ii,k}^d)^4 }{ \sum_{k=1}^{K_{\text{DL}}} (\hat{\rho}^d_{ii,k})^2 }\sqrt{M} \right), \label{eq: R_MF_non_lim}\\
&&\lim_{M\to \infty} R_{d,i}^\text{ZF} \to \sum_{k=1}^{K_\text{DL}} \log_2\left(1+ P_r (\hat{\rho}_{ii,k}^d)^2\sqrt{M} \right). \label{eq: R_ZF_non_lim}
\end{eqnarray}
From (\ref{eq:R_u_MF}) and (\ref{eq:R_u_ZF}), the asymptotic result on the uplink ergodic achievable sum-rate is    
\begin{eqnarray}
&&\lim_{M\to \infty} R_{u,i}^\text{MF} \to \sum_{k=1}^{K_\text{UL}} \log_2\left(1+ \frac{P_r (\hat{\rho}^u_{ii,k})^2}{1+\beta}\sqrt{M} \right), \label{eq: R_MF_lim}\\
&&\lim_{M\to \infty} R_{u,i}^\text{ZF} \to \sum_{k=1}^{K_\text{UL}} \log_2\left(1+ \frac{P_r (\hat{\rho}^u_{ii,k})^2}{1+\beta}\sqrt{M} \right), \label{eq: R_ZF_lim}
\end{eqnarray}
where, without loss of generality, we assume $M_t=M,P_d=P_u$ and $n_0=1$ for all cells. 
\end{lemma}     
\begin{proof}
Due to power scaling in large-scale MIMO \cite{Ngo13TCOM,SimJCN}, we have $P_d=P_u=P_r/\sqrt{M}$. Then, after reorganizing each equation with respect to $M$, we obtain the (\ref{eq: R_MF_non_lim}) - (\ref{eq: R_ZF_lim}). 
\end{proof}

\begin{lemma}[Asymptotic Achievable Sum-rate of Full-Duplex System in Cooperative Multicell]\label{lem: M2}
From (\ref{eq:R_bar_d_MF}) and (\ref{eq:R_bar_d_ZF}), the asymptotic bound on the DL achievable sum-rate is 
\begin{eqnarray}
&&\hspace{-1cm}\lim_{M\to \infty} R_{d,i}^\text{C,MF} \to \\
&&\sum_{k=1}^{K_\text{DL}} \log_2\left(1+ \frac{P_r(1-2^{-C_d})\left(\sum_{j=1}^N (\hat{\rho}_{ij,k}^d)^4 + \sum_{(n,m)\in \Omega_1} (\hat{\rho}_{in,k}^d)^2 (\hat{\rho}_{im,k}^d)^2\right)    \sqrt{M}  }{ \sum_{j=1}^N \left( (\sigma_j^{d,\text{MF}})^2 (\rho_{ij,k}^d)^2 \right) + \sum_{(i,j)=(1,1)}^{(N,N)} \sum_{k=1}^{K_\text{DL}}(\hat{\rho}_{ij,k}^d)^2 } \right), \nonumber\\
&&\hspace{-1cm}\lim_{M\to \infty} R_{d,i}^\text{C,ZF} \to  \sum_{k=1}^{K_\text{DL}} \log_2\left(1+ \frac{P_r(1-2^{-C_d}) (\hat{\rho}_{ik,avg}^d)^2 \sqrt{M} }{M\sum_{j=1}^N \big(\rho_j^{d,ZF})^2(\rho_{ij,k}^d)^2 \big) +1 } \right). 
\end{eqnarray}
From (\ref{eq:R_bar_u_MF}) and (\ref{eq:R_bar_u_ZF}), the asymptotic bound on the uplink achievable sum-rate is 
\begin{eqnarray}
&& \lim_{M\to \infty} R_{u,i}^\text{C, MF} \to \sum_{k=1}^{K_\text{UL}} \log_2\left(1+ \frac{ P_r \sum_{n=1}^N (\hat{\rho}_{in,k}^u)^2 \sqrt{M} }{ \sum_{n=1}^{N} (\sigma_{n}^{u,\text{MF}})^2 (\hat{\rho}_{in,k}^u)^2 / \sum_{n=1}^N (\hat{\rho}_{in,k}^u)^2 + \beta + 1} \right), \\
&& \lim_{M\to \infty} R_{u,i}^\text{C, ZF} \to \sum_{k=1}^{K_\text{UL}} \log_2\left(1+ \frac{P_r \sum_{n=1}^N (\hat{\rho}_{in,k}^u)^2  \sqrt{M}}{ \sum_{n=1}^N (\sigma_n^{u, ZF})^2 / N + \beta + 1} \right),
\end{eqnarray}
where, without loss of generality, we assume $M_t=M,P_d=P_u$ and $n_0=1$ for all cells. 
\end{lemma}
\begin{proof}
Similary, it follows the proof of Lemma \ref{lem: M1}. 
\end{proof}

\begin{remark}[Sum-Rate of Full-Duplex System in Multicell Systems]
From $\mathbf{Lemma\:3}$ and $\mathbf{Lemma\:4}$, for the same given user set, the asymptotic achievable sum-rate of the full-duplex system becomes almost two times greater than those of the half-duplex system even for cell-boundary users, but only if sufficient dynamic range is supported. This is because all interference is mitigated by scaling down the power as $M$ increases. We note, however, that channel estimation error including pilot contamination still affects the limit value. In a cooperative multicell system, additionally, there exists quantization noise in the limit value which linearly depends on $M$ (See $\mathbf{Theorem}$ \ref{theorem: 4} and \ref{theorem: 5}), since more information is required to exchange over the front-haul between the CU and BSs as $M$ increases. Thus, sufficient front-haul capacity, greater than those of non-cooperative multicell systems,  is required to guarantee the sum-rate in a cooperative multicell system. 
\end{remark}
         
\begin{figure}[t]
\centering
\includegraphics[width=17.3cm, height=7.5cm]{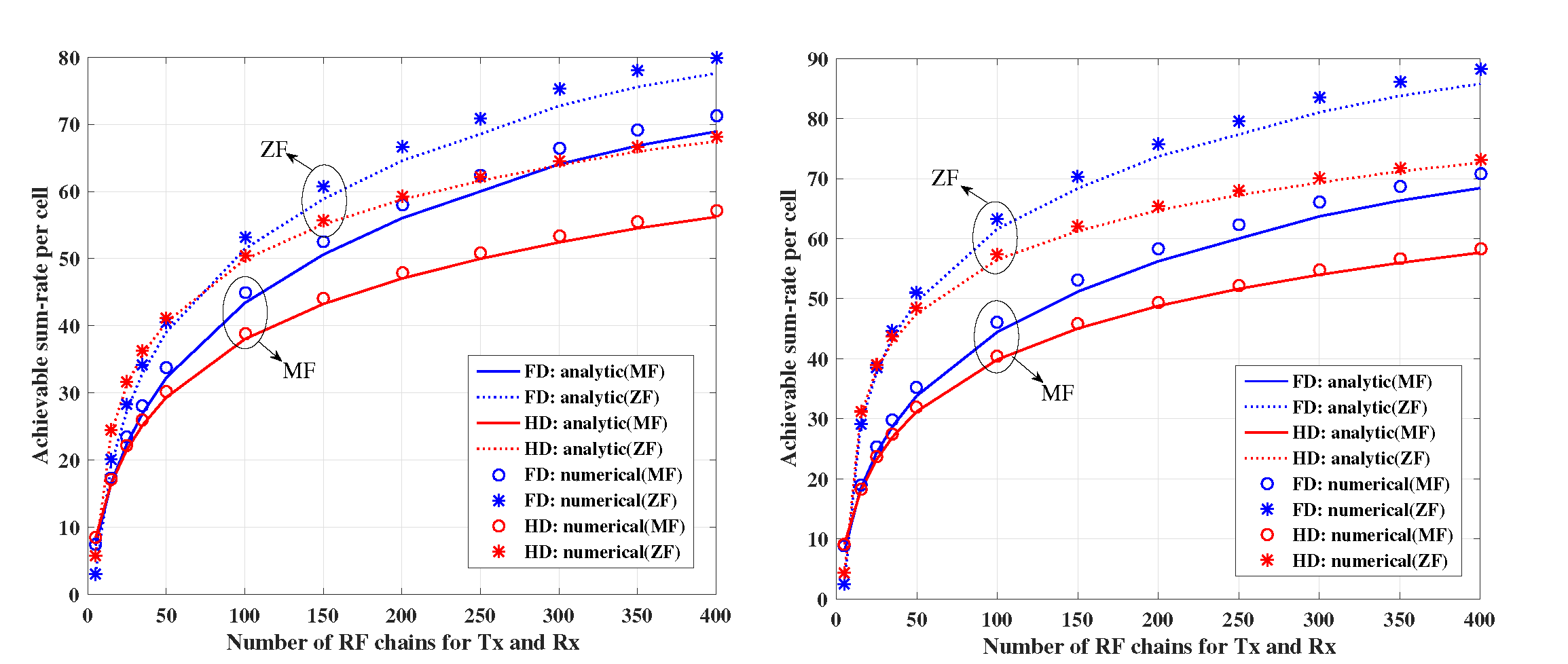}
\caption{Analytic and numerical results of the achievable sum-rate for cell-boundary users. FD and HD stand for  full-duplex and half-duplex system. We consider $N=3, K=5, r=2$km and $P_r=40$dBm. The left figure is for the non-cooperative multicell system;  the right figure is for the cooperative multicell system with $C_u=C_d=20$bps/Hz.} \label{fig: 1_analytic}
\end{figure}
\begin{figure}[t]
\centering
\includegraphics[width=16.8cm, height=7.3cm]{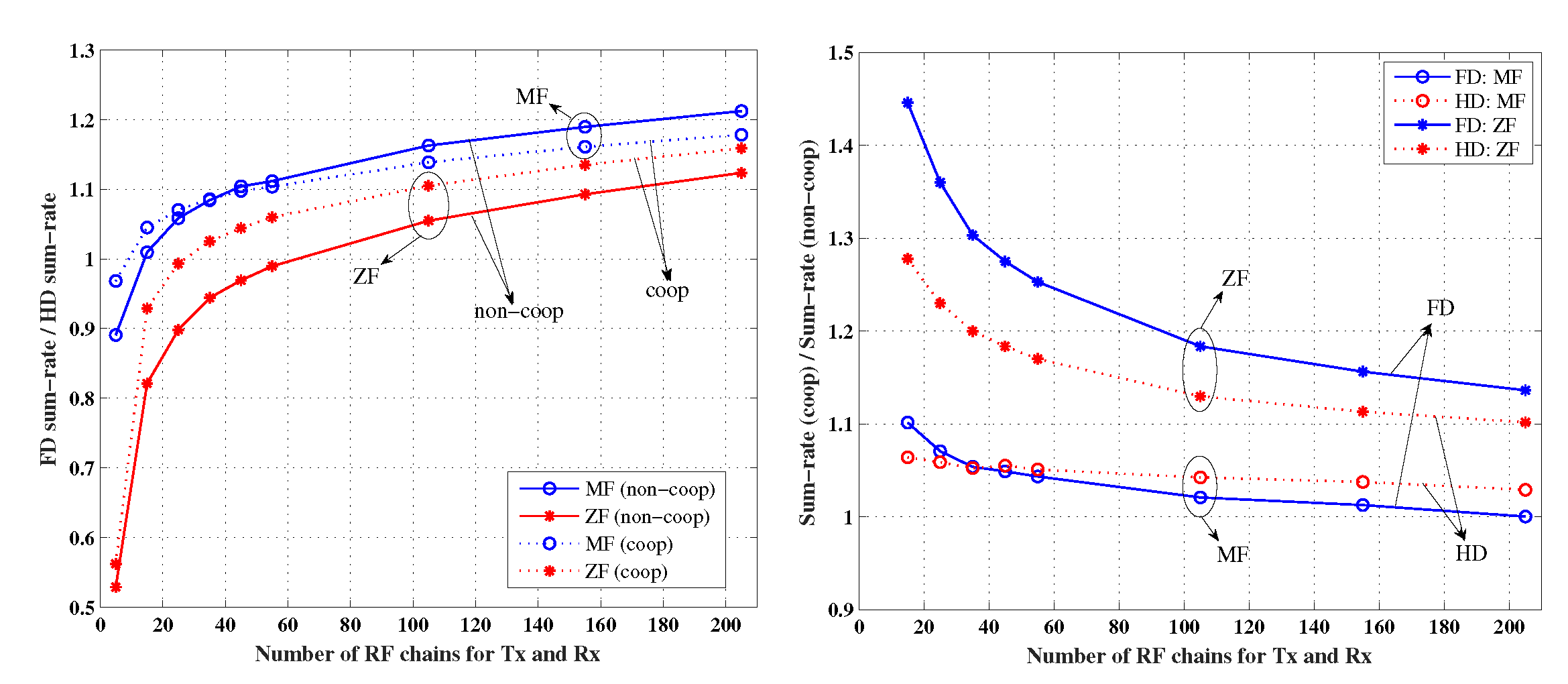}
\caption{Ratio of ergodic achievable sum-rate of full-duplex system to those of  half-duplex system in two different multicell scenarios. We consider the same environment as that shown in Fig. \ref{fig: 1_analytic}.  In the left figure,  coop and non-coop in parentheses stand for cooperative multicell and non-cooperative multicell systems.  } \label{fig: ratio}
\end{figure}
\section{Simulation results} \label{sec: simul}
In this section, we show the results derived in Sections \ref{sec: sum-rate} and  \ref{sec: analysis}, and compare them with the numerical results.  In order to consider cell-boundary users, we consider the users to be distributed around the bottom 5\% of the cell radius based on the LTE standard. We assume that  $M_t=M_r=M$, $K_\text{DL}=K_\text{UL}=K$ and $\alpha = \beta = -100 \text{dB}$ \cite{day2012full}. We assume the pilot overhead to be $\tau_{SI}=M, \tau_\text{UU}=K$, and $\tau_\text{UD}=K$ \cite{yin2013coordinated}.  The  large scale fading coefficient is modeled as $\rho=z/d^v$, where $z$ is a log-normal random variable with variance $\sigma_{\text{shadow}}=$8dB, and $d$ is the distance in meters between the BS and UEs. $v=3.8$ is the pathloss exponent. For the SI channel we assume a free space pathloss $\rho=(4\pi d/ \lambda)^2$, where $\lambda=c/f, c=3\times 10^8 $ and $f=2.4\text{GHz}$.  We define the reference transmission power $P_r$ for which $P_d=P_u=P_r /\sqrt{M}$ \cite{Ngo13TCOM}.  

\begin{figure}[t]
\centering
\includegraphics[width=16.6cm, height=8cm]{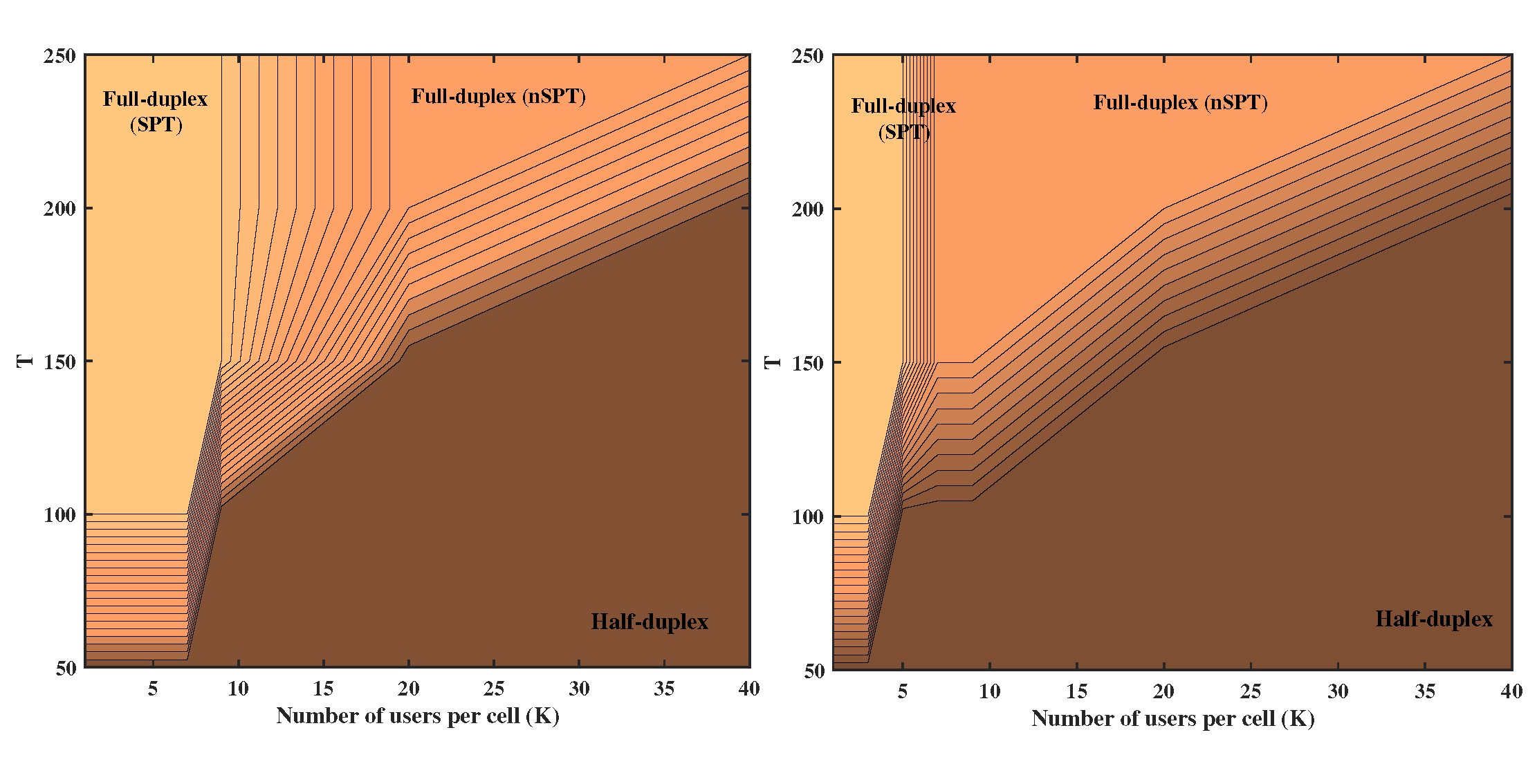}
\caption{Contour plots for comparison of achievable sum-rate of full-duplex system with SPT and nSPT, and the half-duplex system. We set values $N=3, M=50, r=2$km, and $P_r=20$dBm. (a) Non-cooperative multicell system. (b) Cooperative multicell system with $C_u=C_d=20$bps/Hz.  } \label{fig: NoCRAN_KT}
\end{figure}

Figure \ref{fig: 1_analytic}  offers a comparison of the analytic model derived in Section \ref{sec: sum-rate}  with the numerical results for  ZF and MF. In order to focus on the tightness of the analytic model, we consider only nSPT for the full-duplex system and we do not reflect the pilot overhead, since the only difference between nSPT and SPT is the resulting channel estimation error. In Fig. \ref{fig: 1_analytic}, our analytic model can be seen to be well matched with the numerical results. Typically, the performances of MF and ZF depend on the received SINR, which is affected by the radius of the cell or by the reference transmission power.    It is noticeable that the sum-rate of cell-boundary users in full-duplex is enhanced as $M$ increases, even more than those of sum-rate values in the half-duplex system.  This goes against the findings of  previous works \cite{day2012full}, \cite{kam2015bidirectional}. This is because user-user interference and SI can be mitigated because the total transmission power of  the BS and  the transmit power of each user  decrease as the number of antennas at the BS increases, \cite{Rusek13SPMAG, Ngo13TCOM}. On the left figure of Fig. \ref{fig: ratio}, it can be seen that there exists a sum-rate enhancement when using full-duplex operation on the BS side as $M$ increases. Specifically, if the value of the ratio becomes more than 1, the achievable sum-rate of the full-duplex system is larger than those values of  the half-duplex system. This is because the total DL transmit power and the UL transmit power of each user are scaled down as $M$ increases. On the right figure of Fig. \ref{fig: ratio}, we focus on determining the feasibility of cooperation of BSs via front-haul. Specifically, if the value of the  ratio is larger than 1, we are able to achieve more sum-rate via BS cooperation. However, as $M$ increases, we lose the advantage of BS cooperation for the following reasons: i) intercell interference is averaged out due to large number of RF chains even in non-cooperative multicell systems; and ii) cooperation of BSs requires more front-haul capacity for larger number of RF chains and incurs more Tx noise, Rx distortion and self-interference due to joint signal processing over the entire cell at CU.

\begin{figure}[t]
\centering
\includegraphics[width=12cm, height=8cm]{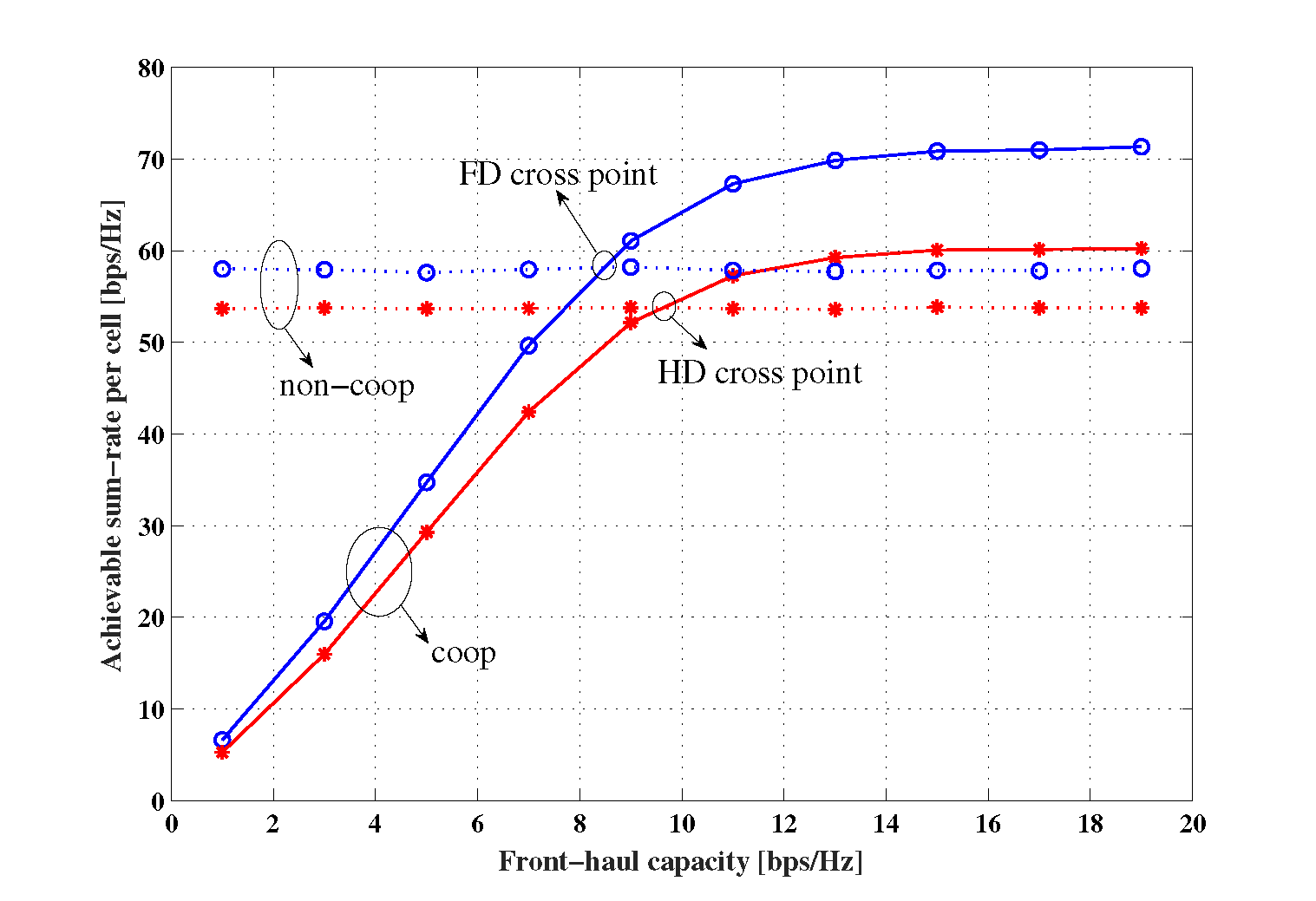}
\caption{Achievable sum-rate with respect to the front-haul capacity. CRAN and NoCRAN stands for the cooperative multicell and non-cooperative multicell systems. 
We consider the values of $N=3, M=128, r=2\text{km}, K=5$ and $P_r=40\text{dBm}$.  } \label{fig: NoCRAN_vs_CRAN}
\end{figure}

In Fig. \ref{fig: NoCRAN_KT}, we compare the achievable sum-rates of the three systems-- full-duplex with nSPT and SPT, and half-duplex systems in non-cooperative multicell and cooperative multicell systems.  In Fig. \ref{fig: NoCRAN_KT}, we show the best-performing system for each regime. For instance, the brighter color in the contour means a superior sum-rate of the full-duplex with SPT, while the darker color implies the superiority of the sum-rate of the half-duplex system. Generally,  both multicell scenarios follow a  similar tendency with respect to $T$ and $K$. With small $T$ which can be regarded as leading to high mobility scenarios or operation at higher frequencies \cite{bjornson2016massive}, the sum-rate of the half-duplex system is better because pilot overhead of $\tau_\text{HD}=2K$ has the smallest value among the three methods.  Then, as $T$ increase, the full-duplex system provides better sum-rate. The advantages of SPT are the reduction of pilot overhead and the provision of  power gain $(\tau_{\text{max}}-\tau_{UU})$, while the UL pilots are interfered  by the  SI pilot. As a result, with large $T$ which can be regarded as indicative of low mobility scenarios or operation at low frequencies, the SPT outperforms nSPT due to additional power gain on UL channel estimation as $K$ decreases. We also observe that, due to the reduction of pilot overhead for SPT,  SPT can even perform better than the half-duplex system  in the regime of small $T$ and $K$, in which the nSPT performs worse than the half-duplex system. 

\begin{figure}[t]
\centering
\includegraphics[width=12cm, height=8.3cm]{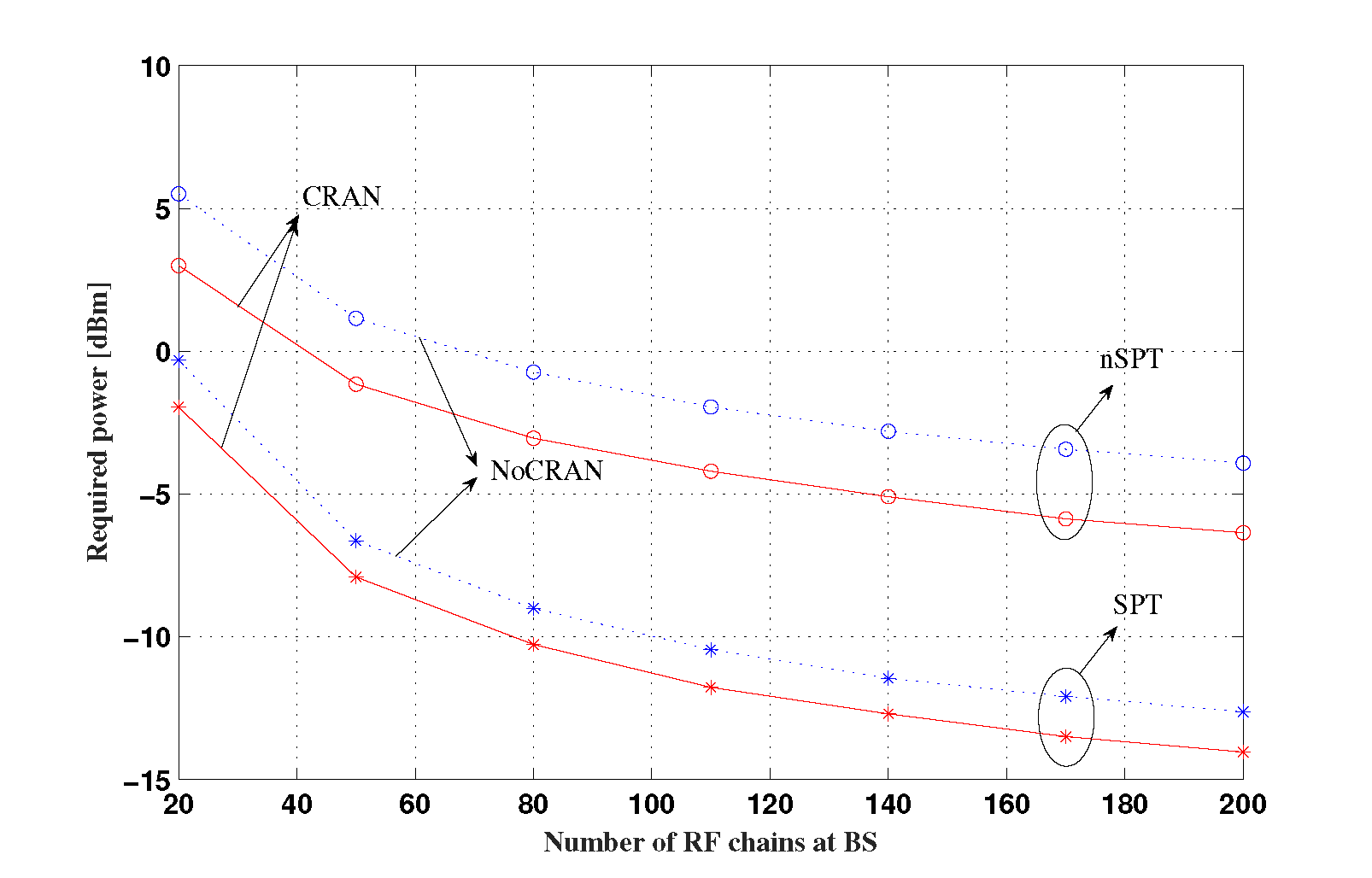}
\caption{Required transmission powers of BS and a user for fixed rate $0.1$bps/Hz. We consider values of $N=3, K=1, r=2\ \text{km}, T=800\ \text{symbols}$ and $C_d=C_u=20\ \text{bps/Hz}$.} \label{fig: R_fixed}
\end{figure}


In Fig. \ref{fig: NoCRAN_vs_CRAN}, we observe the trade-off of BS cooperation and quantization noise induced by limited front-haul capacity. If we have sufficient front-haul capacity, it means here that the front-haul capacity is larger than the cross point, and BS cooperation can be beneficial. In other words, in this regime, the gain obtained by mitigating intercell interference becomes greater than the SINR loss due to quantization noise.   In this context, it is seen that the cooperation of BSs instead decreases the achievable sum-rate as the front-haul capacity decreases.

Figure \ref{fig: R_fixed} shows the required transmission powers of BS and a user when the achievable sum-rate of DL and UL transmissions are fixed. In this context, less required power implies the provision of a better rate under identical power constraints. In line with previous simulation results, the cooperation of BSs requires less power due to intercell interference control only if sufficient front-haul capacity is given. Moreover, under the given environment, the SPT requires less power which means that it provides a higher rate due to the power gain $(\tau_{\max}-\tau_\text{UU})$ attained in estimating the UL channels. 

In Fig. \ref{fig: NMSE}, we observe that the  NMSE of both nSPT and SPT decrease as the transmission power increases, because it is implied that more power is used to send pilots. Furthermore, as discussed in Remark \ref{rem: NMSE}, the NMSE of SPT becomes larger in a very small cell. This is because BS experiences more interference from UL pilot transmission due to the increased large-scale fading gain of UL users. However, even with an increase of the cell radius, we can obtain a similar SI cancellation lever for nSPT by using SPT. 

\begin{figure}[t]
\centering
\includegraphics[width=12cm, height=8.3cm]{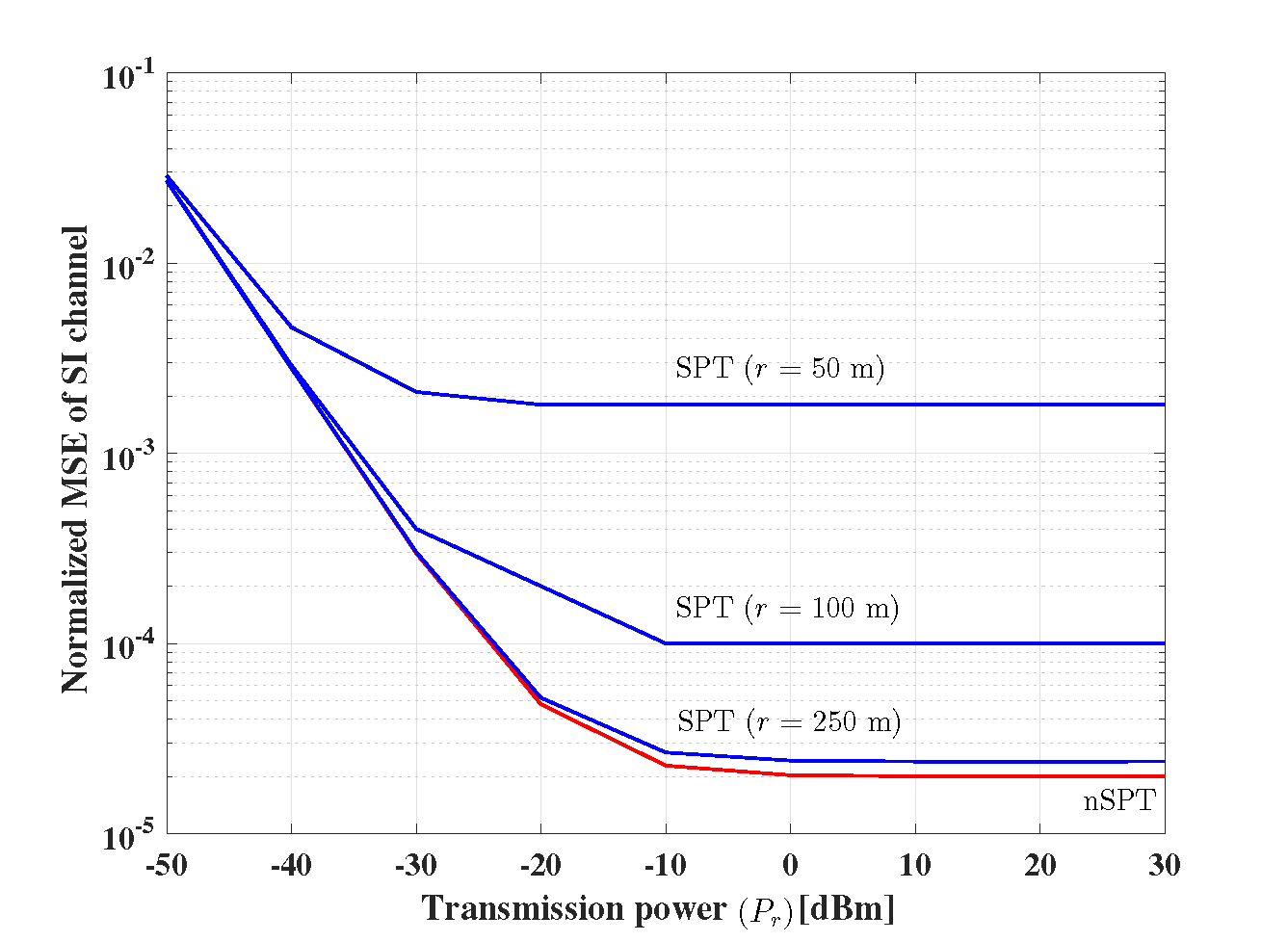}
\caption{NMSE of SI channel for nSPT and SPT with $M=16, N=3, K=5$ and $\alpha=\beta= -50$dB. Moreover, NMSE is shown according to  changing cell radius, $r$. We note that the NMSE of the SI channel for nSPT  depends only on $P_r$ regardless of $r$. } \label{fig: NMSE}
\end{figure}

\section{Conclusion} \label{sec: conclusion}
In this paper, we have investigated the full-duplex large-scale MIMO cellular system while proposing an SPT for channel estimation to resolve the critical pilot overhead problem. We then derived an analytic model considering large-scale fading, pilot contamination, Tx noise and Rx distortion  for the ergodic achievable sum-rate of cell-boundary users in two different multicell scenarios, non-cooperative and cooperative multicell systems. Exploiting the analytic model, we derived the maximum tolerant power of interference and the minimum required coherence time for the reliable region, and investigated the behavior of the asymptotic achievable sum-rate of the full-duplex system. Specifically, as the number of antennas at BS increases, the achievable sum-rate of the full-duplex system is enhanced by means of scaling down of the  transmission power for both multicell systems. However, in the meantime, a cooperative multicell system requires more front-haul capacity to guarantee  higher achievable sum-rate than those of non-cooperative multicell systems.  Using simulation results, we first observed that our analytic model worked well, and our proposed scheme of SPT has advantages of sensitivity to coherence time due to reduced pilot overhead and additional power gain for estimating channels.

\appendices
\section{}
From (\ref{eq: Y_i_SPT}), we estimate the SI channel based on $\tilde{\mathbf{Y}}_i^\text{SPT}=\mathbf{Y}_i^\text{SPT} (\mathbf{\Phi}_i^\text{SI})^*$. The $(\ell,m)^\text{th}$ element of $\tilde{\mathbf{Y}}_i^\text{SPT}$ is
\begin{eqnarray} \label{eq: y_SI}
{\tilde{y}}_{i,\ell m}^\text{SPT}= \sqrt{\tau_{max} P_d} g_{ii,m\ell}^{\text{BS}} + 
\sqrt{\tau_{max} P_u} \sum_{j=1}^N (\mathbf{g}_{ij,m}^{u,r})^{\text{T}} (\mathbf{\Phi}_j^{UL})^\mathrm{T}(\mathbf{\phi}_{i,m}^{SI})^* + n_{i,m\ell} (\mathbf{\phi}_{i,m}^{SI})^* + \\
 \sqrt{\tau_{max} P_u} (\mathbf{g}_{ii,m}^{\text{BS}})^\text{T}(\mathbf{\Psi}_i^\text{SPT})^\mathrm{T}(\mathbf{\phi}_{i,m}^{SI})^* + {\delta}_{i,m\ell}^\text{SPT} (\mathbf{\phi}_{i,m}^{SI})^*. \nonumber
\end{eqnarray}
Since we have  knowledge of the geometric attenuation, we only need to estimate $(\ell,m)^\text{th}$ element of $\mathbf{H}_{ii}^{\text{BS}}$. Based on \cite{sengijpta1995fundamentals}, the estimated channel of  $(\ell,m)^\text{th}$ element of $\mathbf{H}_{ii}^{\text{BS}}$ is
\begin{eqnarray}\label{eq: mmse_eq}
\hat{{h}}_{ii,\ell m}^\text{BS} \hspace{-0.5cm}&&= \mathbf{C}^{\mathrm{T}}_{{{h}}_{ii,\ell m}^\text{BS} {\tilde{y}}_{i,\ell m}^\text{SPT}} \mathbf{C}_{{\tilde{y}}_{i,\ell m}^\text{SPT} {\tilde{y}}_{i,\ell m}^\text{SPT}}^{-1} {\tilde{y}}_{i,\ell m}^\text{SPT} + \mathbb{E} \big[ {{h}}_{ii,\ell m}^\text{BS} \big] - \mathbf{C}_{{{h}}_{ii,\ell m}^\text{BS} {\tilde{y}}_{i,\ell m}^\text{SPT}} \mathbf{C}_{{\tilde{y}}_{i,\ell m}^\text{SPT} {\mathbf{y}}_{i,m}^\text{SPT}} \mathbb{E}\big[ {\tilde{y}}_{i,\ell m}^\text{SPT}  \big] \label{eq: h_hat_SI} \\
&&\overset{(a)} = \mathbf{C}^{\mathrm{T}}_{{{h}}_{ii,\ell m}^\text{BS} {\tilde{y}}_{i,\ell m}^\text{SPT}} \mathbf{C}_{{\tilde{y}}_{i,\ell m}^\text{SPT} {\tilde{y}}_{i,\ell m}^\text{SPT}}^{-1} {\tilde{y}}_{i,\ell m}^\text{SPT}, \nonumber
\end{eqnarray}
(a) is due to the fact that zero mean of the second and third terms. From (\ref{eq: mmse_eq}), we obtain 
\begin{eqnarray}
&&\hspace{-0.5cm}\hat{{\mathbf{h}}}_{ii,\ell m}^\text{BS}= \\
&&\sqrt{\tau_{\max}P_d}\left(  P_u \sum_{j=1}^N \sum_{k=1}^{K_\text{UL}} (\rho_{ij,k}^{u})^2 +(1+\beta) \left\{\tau_{\max} P_d (\rho_{ii,\ell m}^\text{BS})^2 + \alpha P_d \sum_{m=1}^{M_t} (\rho_{ii,\ell m}^\text{BS})^2 + n_0\right\} \right)^{-1} {\tilde{y}}_{i,\ell m}^\text{SPT}. \nonumber
\end{eqnarray}
In a similar manner, we also derive the estimated channel of the $k^\text{th}$ column of $\mathbf{H}_{i,i}^u$ which is described as 
\begin{eqnarray}
\hat{\mathbf{h}}_{ii,k}^u= \sqrt{\tau_\text{max}P_u} \left( \left\{ \tau_{\max}P_u \sum_{j=1}^N (\rho_{ij,k}^{u})^2 + P_d \sum_{m=1}^{M_t} (\rho_{ii,m}^\text{BS})^2\mathrm{Var}[\tilde{\epsilon}_{ii,m}^\text{BS}] +  n_o +A \right\} \mathbf{I}_{M_r} \right) ^{-1} \mathbf{y}_{i,k}^{\text{SPT},u},
\end{eqnarray}
where $A=\beta\{P_d\sum_{m=1}^{M_t}(\rho_{ii,m}^\text{BS})^2+\alpha P_d\sum_{m=1}^{M_t}(\rho_{ii,m}^\text{BS})^2  + n_o \}$. $\Box$

\section{}
Based on \cite[Sec. III-F]{lim2015performance}, we obtain   $\mathbb{E} \left[\log_2 (1+\text{SINR}) \right] \approx \log_2 (1+ \mathbb{E} \left[\text{SINR} \right] )$. 
For the MF precoder, after multiplying $||\hat{\mathbf{G}}_{ii}^d||^2$ to the denominator and numerator of SINR of (\ref{eq:y_NoCRAN}), 
\begin{subequations}
\begin{eqnarray}
\mathbb{E} \left[\text{SINR}_{d,i,k}^\text{MF} \right ]\hspace{-0.7cm} &&=  \mathbb{E} \left[  \frac{P_d ||\hat{\mathbf{g}}_{ii,k}^d||^4 } {I_{d,i,k}^\text{MF}+ I_{\text
{UE},d,i,k}^\text{MF} + P_d ||(\mathbf{\epsilon}_{ii,k}^d)^\text{T} \mathbf{g}_{ii,k}||^2 + ||\hat{\mathbf{G}}_{ii}^d||^2 }  \right] \\ &&\overset{(a)}\approx   \frac{ \mathbb{E} \big[ P_d ||\hat{\mathbf{g}}_{ii,k}^d||^4 \big]} {\mathbb{E} \big[\mathrm{I_{d,i,k}^{MF}}+ \mathrm{I_{UE,d,i,k}^{MF}} + P_d ||(\mathbf{\epsilon}_{ii,k}^d)^\text{T} \mathbf{g}_{ii,k}||^2 + ||\hat{\mathbf{G}}_{ii}^d||^2 \big] },    \nonumber
\end{eqnarray}
\begin{eqnarray}
&&{I_{d,i,k}^\text{MF}}=P_d \sum_{l=1,l \neq k}^{K_{DL}} ||(\mathbf{g}_{ii,k}^d)^\text{T} (\hat{\mathbf{g}}_{ii,l}^d)^*||^2 + P_d ||\hat{\mathbf{G}}_{ii}^d||^2 \sum_{j=1,j\neq i}^N \frac{ || (\mathbf{g}_{ij,k}^d)^\text{T} (\hat{\mathbf{G}}_{jj}^d)^H ||^2}{||\hat{\mathbf{G}}_{jj}^d||^2 },
\end{eqnarray}
\end{subequations}
where ${I_{d,i}^\text{MF}}={I_{d,i,k}^\text{MF}}+{I_{\text{UE},d,i,k}^\text{MF}}, {I_{\text{UE},d,i,k}^\text{MF}}=P_u ||\hat{\mathbf{G}}_{ii}^d||^2 \sum_{j=1}^N ||\mathbf{g}_{ij,k}^\text{UE}||^2$ from (\ref{eq:R_d_MF}). (a) follows Lemma 4 in \cite{lim2015performance}. Based on $||\mathbf{AB}||^2 = \text{tr}(\mathbf{ABBA})$, we calculate each terms. For ZF precoder, we exploit the result $\mathbb{E}\big [ \frac{1}{||\mathbf{f}_{i,k}||^2} \big ] = \hat{\rho}_{ii,k}^2(M_t-K_\text{DL}+1)$, which is expanded from \cite{wong2008array}. Then, the ergodic SINR is 
\begin{eqnarray}
\mathbb{E} \big[\text{SINR}_{d,i,k}^\text{ZF} \big ] \hspace{-0.7cm} &&  \approx \\
&&\hspace{-2.3cm}\frac{\mathbb{E} \big[ P_d \frac{1}{K_\text{DL} ||\mathbf{f}_{i,k}||^2} \big ]}{\mathbb{E} \big[P_d\sum_{j=1, j\neq i}^N ||(\mathbf{h}_{ij,k}^d)^T \mathbf{F}_j ||^2 + \sum_{l=1, l\neq k}^{K_\text{DL}} || \mathbf{\epsilon}_{ii,k}^T \mathbf{f}_{i,l}||^2+ P_u \sum_{j=1}^N ||\mathbf{g}_{ij,k}^d||^2  + P_d ||(\mathbf{\epsilon}_{ii,k}^d)^\text{T} \mathbf{f}_{i,k}||^2 + n_0 \big]. } \nonumber 
\end{eqnarray}
We note that the residual intracell interference exists due to channel estimation error. Using the property of vector normalization, we compute each term. For instance,  $\mathbb{E}||(\mathbf{h}_{ij,k}^d)^T \mathbf{F}_j ||^2 = \mathbb{E}\big[\text{tr}\big((\mathbf{F}_j)^H (\mathbf{h}_{ij,k}^d)^*(\mathbf{h}_{ij,k}^d)^T \mathbf{F}_j \big)\big]=(\rho_{ij,k}^d)^2 \text{E}\big[\text{tr}\big((\mathbf{F}_j)^H \mathbf{F}_j \big)\big]= (\rho_{ij,k}^d)^2$. $\Box$

\section{}
In a similar manner of Appendix B, we first calculate the UL ergodic SINR for the MF detection filter after multiplying $\frac{1}{||\hat{\mathbf{g}}_{ii,k}^u||^2}$ to both the denominator and numerator. We obtain
\begin{subequations}
\begin{eqnarray}
&&\hspace{-1.5cm}\mathbb{E} \big[ \text{SINR}_{u,i,k}^\text{MF} \big] \approx \\
&& \frac{\mathbb{E} \big[ P_u ||\hat{\mathbf{g}}_{ii,k}^u||^2 \big]}{\mathbb{E}\big[ {I_{u,i,k}^\text{MF}} + P_u||(\hat{\mathbf{g}}_{ii,k}^u)^H \mathbf{\epsilon}_{ii,k}^u||^2 + ||\hat{\mathbf{g}}_{ii,k}^u\mathbf{n}_{ui,k}||^2 + ||(\hat{\mathbf{g}}_{ii,k}^u)^H \mathbf{\delta}_i||^2+ || (\hat{\mathbf{g}}_{ii,k}^u)^H \mathbf{n}_{ui,k}||^2 \big] }, \nonumber\\ 
&&\hspace{-1.5cm}{I_{u,i,k}^\text{MF}}=P_u \sum_{l=1,l\neq k}^{K_{UL}} ||(\hat{\mathbf{g}}_{ii,k}^u)^H \mathbf{g}_{ii,l}^u||^2 + P_u \sum_{j=1, j\neq i}^N ||(\hat{\mathbf{g}}_{ii,k}^u)^H \mathbf{G}_{jj}^u ||^2 \nonumber \\ 
&&+ P_d \sum_{j=1, j\neq i}^N \left|\left| (\hat{\mathbf{g}}_{ii,k}^u)^H \mathbf{G}_{ij}^\text{BS} \frac{(\hat{\mathbf{G}}_{jj}^d)^H}{||\hat{\mathbf{G}}_{jj}^d||}\right|\right|^2 + 
P_d \left|\left| (\hat{\mathbf{g}}_{ii,k}^u)^H (\mathbf{G}_{ii}^\text{BS}-\hat{\mathbf{G}}_{ii}^\text{BS}) \frac{(\hat{\mathbf{G}}_{ii}^d)^H}{||\hat{\mathbf{G}}_{ii}^d||}\right|\right|^2.  
\end{eqnarray}
\end{subequations}
For the ZF precoder, we use $\mathbb{E}\big[ ||\mathbf{w}_{i,k}||^2 \big] = \frac{1}{(\hat{\rho}_{ii,k}^u)^2 (M_r-K_\text{UL}+1)}$. Since the remain procedures are similar, we omit the details. $\Box$

\section{}
Since we produce a precoder based on the global CSI, the collective channel vector includes elements of adjacent cells which follows different variance. For the MF precoder, we calculate those term with modification as follows:
\begin{eqnarray}
\mathbb{E} \big[ P_s || \hat{\mathbf{g}}_{i,k}^d ||^4 \big]  =& \hspace{-1.4cm} P_s \mathbb{E} \big[ \big(||\hat{\mathbf{g}}_{i,k}^d||^2 \big)^2 \big] = P_s \mathbb{E} \big[ \big( ||\hat{\mathbf{g}}_{i1,k}^d||^2 + \ldots + ||\hat{\mathbf{g}}_{iN,k}^d||^2 \big)^2 \big] \\
\overset{(a)} =& \hspace{-1.3cm} P_s \mathbb{E} \big [ \sum_{j=1}^N ||\hat{\mathbf{g}}_{ij,k}^d||^2 \big] + P_s \mathbb{E} \big [  \hspace{-.5mm} \sum_{(n,m) \in \Omega_1} \hspace{-2mm} ||\hat{\mathbf{g}}_{in,k}^d||^2 
||\hat{\mathbf{g}}_{im,k}^d||^2 \big] \nonumber \\
\overset{(b)} =& \hspace{-.2cm} P_s \sum_{j=1}^N M_t(M_t+1) \big( \hat{\rho}_{ij,k}^d \big)^4 + P_s\sum_{(n,m) \in \Omega_1} M_t^2 \big(\hat{\rho}_{in,k}^d \big)^2 \big(\hat{\rho}_{im,k}^d \big)^2. \nonumber
\end{eqnarray}
Since all elements are independent, (a) is obtained by developing the equation, where $\Omega_1$ is defined in the previous section. (b) follows the result after averaging each elements out based on  Lemma 1 in \cite{lim2015performance}. For the ZF precoder, we calculate as 
\begin{eqnarray}
\mathbb{E} || (\mathbf{\epsilon}_{i,k})^\mathrm{T} \mathbf{f}_{i,k}||^2 & \hspace{-3.8cm}= \mathbb{E}\big[\mathrm{tr}(\mathbf{f}_{i,k}^\dagger \mathbf{\epsilon}_{i,k}^* \mathbf{\epsilon}_{i,k}^\mathrm{T} \mathbf{f}_{i,k})] \\
&\overset{(a)} \approx (\bar{\rho}_{ik,avg}^d)^2 \mathrm{E}\mathrm{tr}(\mathbf{f}_{i,k}^\dagger \mathbf{f}_{i,k}) 
= (\bar{\rho}_{ik,avg}^d)^2 \frac{1}{K_{DL} N }, \nonumber
\end{eqnarray}
where $(\bar{\rho}_{ik,\text{avg}}^d)^2= \mathbb{E}\big[(\bar{\rho}_{i1,k}^d)^2 \ldots (\bar{\rho}_{iN,k}^d)^2 \big]$. In (a),  $\mathbb{E}[\mathbf{\epsilon}_{i,k}^* \mathbf{\epsilon}_{i,k}^{T}]= \mathrm{diag}\big[ (\bar{\rho}_{i1,k}^d)^2 \mathbf{I}_{M_t} \ldots (\bar{\rho}_{iN,k}^d)^2 \mathbf{I}_{M_t} \big]$ $\approx (\bar{\rho}_{ik,\text{avg}}^d)^2 \mathbf{I}_{M_t N} $. Based on \cite{sengijpta1995fundamentals}, the desired signal term  
is bounded as $\hat{\rho}_{\min}^2 (M_t - K_\text{DL} +1) \leq \frac{1}{||\mathbf{f}_{i,k}||^2} \leq \hat{\rho}_{\max}^2 (M_t - K_\text{DL} +1)$, where $\hat{\rho}_{\min}^2  (\hat{\rho}_{\max}^2)$ is $ \min (\max) \big[ (\hat{\rho}_{i1,k}^d)^2 \ldots (\hat{\rho}_{iN,k}^d)^2 \big]$. Then, we approximate $\frac{1}{||\mathbf{f}_{i,k}||^2}$ as $\hat{\rho}_\text{avg}^2 (M_t - K_\text{DL} +1)$, where $\hat{\rho}_\text{avg}^2= \mathbb{E}  \big[ (\hat{\rho}_{i1,k}^d)^2 \ldots (\hat{\rho}_{iN,k}^d)^2 \big]$. The ergodic SINR is obtained based on same mathematical background as previous appendices.  $\Box$

\section{}
As we discussed in Appendix D, the channel vector contains elements with different variance. Here, we describe the most complicated term to calculate as follows: 
\begin{subequations}
\begin{eqnarray}
 && \hspace{-.7cm}\mathbb{E} \left[ ||\left(\mathbf{g}_{i,k}^u\right)^H \mathbf{G}^\text{BS}_\text{off}\left(\mathbf{G}^d\right)^H||^2 \right] \hspace{-1cm} \\
&&\hspace{-0.7cm}= \mathbb{E} \left[
\mathrm{tr} \left(\mathbf{G}^d \left(\mathbf{G}^\text{BS}_\text{off}\right)^H \mathbf{g}_{j,k}^u \left(\mathbf{g}_{j,k}^u\right)^H \mathbf{G}^\text{BS}_\text{off}\left(\mathbf{G}^d\right)^H \right) \right] = \mathbb{E} \left[\mathrm{tr} \left(\left(\mathbf{G}^\text{BS}_\text{off}\right)^H \mathbf{g}_{j,k}^u \left(\mathbf{g}_{j,k}^u\right)^H \mathbf{G}^\text{BS}_\text{off}\left(\mathbf{G}^d\right)^H \mathbf{G}^d \right) \right] \nonumber \\
&&\hspace{-0.7cm}\overset{(a)}= \mathrm{tr} \Bigg( \mathrm{blkdiag} \left[\left(\sum_{n=2}^N \left(\hat{\rho}_{in,k}^u\right)^2 \left(\rho^\text{BS}_{n,1}\right)^2 \right) \mathbf{I}_{M_t}, \ldots,  \left(\sum_{n=1}^{N-1} \left(\hat{\rho}_{in,k}^u\right)^2 \left(\rho^\text{BS}_{n,N}\right)^2  \right)  \mathbf{I}_{M_t} \right]  \\ \nonumber  
&&\hspace{4cm}\times \mathrm{blkdiag}\left[\left( \sum_{n=1}^N \sum_{k=1}^{K_{DL}} (\hat{\rho}_{n1,k}^d)^2 \right) \mathbf{I}_{M_t}, \ldots, \left( \sum_{n=1}^N \sum_{k=1}^{K_{DL}} (\hat{\rho}_{nN,k}^d)^2 \right) \mathbf{I}_{M_t} \right]\Bigg) \nonumber \\  
&&\hspace{-0.7cm}= \sum_{j=1}^N \left \{ \left(\sum_{n=1}^N \left(\hat{\rho}_{ik,n}^u\right)^2 \left(\rho^\text{BS}_{n,j}\right)^2  \right) \left(\sum_{n=1}^N \sum_{k=1}^{K_{DL}} \left(\hat{\rho}_{ni,k}^d\right)^2 \right) \right\} M_t, \hspace{2cm} \nonumber
\end{eqnarray}
\end{subequations}

$(a)$ follows the result from calculating $(\mathbf{G}^\text{BS}_\text{off})^H \mathbf{g}_{j,k}^u (\mathbf{g}_{j,k}^u)^H \mathbf{G}^\text{BS}_\text{off}$ and $ (\mathbf{G}^d)^H \mathbf{G}^d $, respectively. For the ZF detection filter, $||\mathbf{w}_{i,k}^T \mathbf{G}^\text{BS}_\text{off} \mathbf{F}^\text{C}||^2 =\mathbb{E} \big[ \mathrm{tr} \big(\mathbf{w}_{i,k}^{T} \mathbf{G}^\text{BS}_\text{off} \mathbf{F}^C \mathbf{F}^H (\mathbf{G}^\text{BS}_\text{off})^H \mathbf{w}_{i,k}^* \big) \big]$ is also the most complicated term to calculate whereby $ \mathbb{E} \big[ \mathbf{F}^\text{C}\mathbf{F}^H \big] = \frac{1}{M_t N} \mathbf{I}_{M_t N} $ due to vector-normalization. And, $\mathbb{E} \big[ \mathbf{G}^\text{BS}_\text{off}(\mathbf{G}_\text{off}^\text{BS})^H \big]$ becomes block-diagonal matrix whose block-diagonal element is $g_i=\sum_{j=1, j \neq i}^N M_t (\rho_{ji}^\text{BS})^2$ for $i=1 \ldots N$. Thus, we obtain the approximation $\mathbb{E} \big[ \mathbf{G}^\text{BS}_\text{off}(\mathbf{G}^\text{BS}_\text{off})^H \big] \approx g_{avg} \mathbf{I}_{M_t N} $.
With the result of $\mathbb{E}\big[ ||\mathbf{w}_{i,k}^{T}||^2 \big] \approx \frac{1}{(\hat{\rho}_{ik,\text{avg}}^u)^2(M_rN - K_\text{UL}N +1)}$, discussed in Appendix C, we calculate $P_s ||\mathbf{w}_{i,k}^{T} \mathbf{G}^\text{BS}_\text{off} \mathbf{F}^\text{C} ||^2$. For the remaining terms, we calculate based on the mathematical background that was used in the preceding appendices. $\Box$

\bibliographystyle{IEEEtran}
\bibliography{refKJW}

\begin{thebibliography}{10}
\providecommand{\url}[1]{#1}
\csname url@samestyle\endcsname
\providecommand{\newblock}{\relax}
\providecommand{\bibinfo}[2]{#2}
\providecommand{\BIBentrySTDinterwordspacing}{\spaceskip=0pt\relax}
\providecommand{\BIBentryALTinterwordstretchfactor}{4}
\providecommand{\BIBentryALTinterwordspacing}{\spaceskip=\fontdimen2\font plus
\BIBentryALTinterwordstretchfactor\fontdimen3\font minus
  \fontdimen4\font\relax}
\providecommand{\BIBforeignlanguage}[2]{{%
\expandafter\ifx\csname l@#1\endcsname\relax
\typeout{** WARNING: IEEEtran.bst: No hyphenation pattern has been}%
\typeout{** loaded for the language `#1'. Using the pattern for}%
\typeout{** the default language instead.}%
\else
\language=\csname l@#1\endcsname
\fi
#2}}
\providecommand{\BIBdecl}{\relax}
\BIBdecl

\bibitem{Gesbert2007shifting}
D.~Gesbert, M.~Kountouris, R.~W. Heath~Jr, C.-B. Chae, and T.~S{\"a}lzer,
  ``Shifting the {MIMO} paradigm,'' \emph{IEEE Sig. Proc. Mag.}, vol.~24,
  no.~5, pp. 36--46, Oct. 2007.

\bibitem{itu2015}
{ITU-R}, ``Rep. {ITU-R M}. 2370-0; {IMT} traffic estimates for the years 2020
  to 2030,'' July 2015.

\bibitem{Jeff14what}
J.~G. Andrews, S.~Buzzi, W.~Choi, S.~V. Hanly, A.~Lozano, A.~C.~K. Soong, and
  J.~C. Zhang, ``What will {5G} be?'' \emph{IEEE J. Sel. Areas Commun.},
  vol.~32, no.~6, pp. 1065--1082, June 2014.

\bibitem{Boccardi14five}
F.~Boccardi, {R. W. Heath, Jr.}, A.~Lozano, T.~L. Marzetta, and P.~Popovski,
  ``Five disruptive technology directions for {5G},'' \emph{IEEE Commun. Mag.},
  vol.~52, no.~2, pp. 74--80, Feb. 2014.

\bibitem{Osseiran2014scenarios}
A.~Osseiran, F.~Boccardi, V.~Braun, K.~Kusume, P.~Marsch, M.~Maternia,
  O.~Queseth, M.~Schellmann, H.~Schotten, and H.~Taoka, ``Scenarios for {5G}
  mobile and wireless communications: the vision of the {METIS} project,''
  \emph{IEEE Commun. Mag.}, vol.~52, no.~5, pp. 26--35, May 2014.

\bibitem{sim2017nonlinear}
M.~S. Sim, M.~Chung, D.~Kim, J.~Chung, D.~K. Kim, and C.-B. Chae, ``Nonlinear
  self-interference cancellation for full-duplex radios: From link-level and
  system-level performance perspectives,'' \emph{to appear in IEEE Commun.
  Mag.}, 2017.

\bibitem{kim2017asymmetric}
S.~Kim, M.~S. Sim, C.-B. Chae, and S.~Choi, ``Asymmetric simultaneous transmit
  and receive in {W}i{F}i networks,'' \emph{IEEE Access}, vol.~5, pp.
  14\,079--14\,094.

\bibitem{chung2017compact}
M.~Chung, M.~S. Sim, D.~Kim, and C.-B. Chae, ``Compact full duplex {MIMO}
  radios in {D2D} underlaid cellular networks: From system design to prototype
  results,'' \emph{to appear in IEEE Access}, 2017.

\bibitem{Choi2010achieving}
J.~I. Choi, M.~Jain, K.~Srinivasan, P.~Levis, and S.~Katti, ``Achieving single
  channel, full duplex wireless communication,'' in \emph{Proc. ACM Mobile
  Computing and Networking}, Sep. 2010, pp. 1--12.

\bibitem{Bharadia2013full}
D.~Bharadia, E.~McMilin, and S.~Katti, ``Full duplex radios,'' in \emph{Proc.
  ACM SIGCOMM Computer Commun. volume={43}, number={4}, pages={375--386},
  year={Aug. 2013}, organization={ACM}}.

\bibitem{chung2015prototyping}
M.~Chung, M.~S. Sim, J.~Kim, D.~K. Kim, and C.-B. Chae, ``Prototyping real-time
  full duplex radios,'' \emph{IEEE Commun. Mag.}, vol.~53, no.~9, pp. 56--63,
  Sep. 2015.

\bibitem{Karakus2015opportunistic}
C.~Karakus and S.~Diggavi, ``Opportunistic scheduling for full-duplex
  uplink-downlink networks,'' in \emph{Proc. IEEE Int. Symp. Info. Theory
  (ISIT)}, Oct. 2015, pp. 1019--1023.

\bibitem{Lee2015hybrid}
J.~Lee and T.~Q. Quek, ``Hybrid full-/half-duplex system analysis in
  heterogeneous wireless networks,'' \emph{IEEE Trans. Wireless Commun.},
  vol.~14, no.~5, pp. 2883--2895, Jan. 2015.

\bibitem{goyal2015full}
S.~Goyal, P.~Liu, S.~S. Panwar, R.~A. DiFazio, R.~Yang, and E.~Bala, ``Full
  duplex cellular systems: will doubling interference prevent doubling
  capacity?'' \emph{IEEE Commun. Mag.}, vol.~53, no.~5, pp. 121--127, May 2015.

\bibitem{simeone2014full}
O.~Simeone, E.~Erkip, and S.~Shamai, ``Full-duplex cloud radio access networks:
  An information-theoretic viewpoint,'' \emph{IEEE Wireless Lett.}, vol.~3,
  no.~4, pp. 413--416, Aug. 2014.

\bibitem{yin2013full}
B.~Yin, M.~Wu, C.~Studer, J.~R. Cavallaro, and J.~Lilleberg, ``Full-duplex in
  large-scale wireless systems,'' in \emph{Proc. of Asilomar Conf. on Sign.,
  Syst. and Computers (ASILOMAR)}, Nov. 2013, pp. 1623--1627.

\bibitem{lim2015performance2}
Y.-G. Lim, D.~Hong, and C.-B. Chae, ``Performance analysis of self-interference
  cancellation methods in full-duplex large-scale {MIMO} systems,'' \emph{Proc.
  IEEE Glob. Commun. Conf. (GLOBECOM)}, Dec. 2016.

\bibitem{Rusek13SPMAG}
F.~Rusek, D.~Persson, B.~K. Lau, E.~G. Larsson, T.~L. Marzetta, O.~Edfors, and
  F.~Tufvesson, ``Scaling up {MIMO}: Opportunities and challenges with very
  large arrays,'' \emph{IEEE Sig. Proc. Mag.}, vol.~30, no.~1, pp. 40--60, Jan.
  Jan. 2013.

\bibitem{Ngo13TCOM}
H.~Q. Ngo, E.~G. Larsson, and T.~L. Marzetta, ``Energy and spectral efficiency
  of very large multiuser {MIMO} systems,'' \emph{IEEE Trans. Commun.},
  vol.~61, no.~4, pp. 1436--1449, Feb. 2013.

\bibitem{day2012full}
B.~P. Day, A.~R. Margetts, D.~W. Bliss, and P.~Schniter, ``Full-duplex
  bidirectional {MIMO}: Achievable rates under limited dynamic range,''
  \emph{IEEE Trans. Sig. Proc.}, vol.~60, no.~7, pp. 3702--3713, Apr. 2012.

\bibitem{lim2015performance}
Y.-G. Lim, C.-B. Chae, and G.~Caire, ``Performance analysis of massive {MIMO}
  for cell-boundary users,'' \emph{IEEE Trans. Wireless Commun.}, vol.~14,
  no.~12, pp. 6827--6842, Jul. 2015.

\bibitem{yun2016intra}
J.-H. Yun, ``Intra and inter-cell resource management in full-duplex
  heterogeneous cellular networks,'' \emph{IEEE Trans. on Mobile Comput.},
  vol.~15, no.~2, pp. 392--405, Feb. 2016.

\bibitem{goyal2014improving}
S.~Goyal, P.~Liu, S.~Panwar, R.~A. DiFazio, R.~Yang, J.~Li, and E.~Bala,
  ``Improving small cell capacity with common-carrier full duplex radios,'' in
  \emph{Proc. IEEE Int. Conf. on Commun. (ICC)}, Jun. 2014, pp. 4987--4993.

\bibitem{Jose11TWC}
J.~Jose, A.~Ashikhmin, and T.~L. Marzetta, ``Pilot contamination and precoding
  in multi-cell {TDD} systems,'' \emph{IEEE Trans. Wireless Commun.}, vol.~10,
  no.~8, pp. 2640--2651, Aug. 2011.

\bibitem{kang2014joint}
J.~Kang, O.~Simeone, J.~Kang, and S.~S. Shitz, ``Joint signal and channel state
  information compression for the backhaul of uplink network {MIMO} systems,''
  \emph{IEEE Trans. Wireless Commun.}, vol.~13, no.~3, pp. 1555--1567, Jan.
  2014.

\bibitem{suzuki2008transmitter}
H.~Suzuki, T.~V.~A. Tran, I.~B. Collings, G.~Daniels, and M.~Hedley,
  ``Transmitter noise effect on the performance of a {MIMO}-{OFDM} hardware
  implementation achieving improved coverage,'' \emph{IEEE J. Sel. Areas
  Commun.}, vol.~26, no.~6, pp. 867--876, Aug. 2008.

\bibitem{sengijpta1995fundamentals}
S.~K. Sengijpta, ``Fundamentals of statistical signal processing: Estimation
  theory,'' \emph{Technometrics}, vol.~37, no.~4, pp. 465--466, 1995.

\bibitem{SimJCN}
M.~S. Sim, J.~Park, C.-B. Chae, and {R. W. Heath, Jr.}, ``Compressed channel
  feedback for correlated massive {MIMO} systems,'' \emph{IEEE/KICS J. Commun.
  Netw.}, vol.~18, no.~1, pp. 95--104, Feb. 2016.

\bibitem{yin2013coordinated}
H.~Yin, D.~Gesbert, M.~Filippou, and Y.~Liu, ``A coordinated approach to
  channel estimation in large-scale multiple-antenna systems,'' \emph{IEEE J.
  Sel. Areas Commun.}, vol.~31, no.~2, pp. 264--273, Jan. 2013.

\bibitem{kam2015bidirectional}
S.~Kam, D.~Kim, H.~Lee, and D.~Hong, ``Bidirectional full-duplex systems in a
  multispectrum environment,'' \emph{IEEE Trans. Veh. Technol.}, vol.~64,
  no.~8, pp. 3812--3817, Aug. 2015.

\bibitem{bjornson2016massive}
E.~Bj{\"o}rnson, E.~G. Larsson, and T.~L. Marzetta, ``Massive {MIMO}: Ten myths
  and one critical question,'' \emph{IEEE Commun. Mag.}, vol.~54, no.~2, pp.
  114--123, Feb. 2016.

\bibitem{wong2008array}
K.-K. Wong and Z.~Pan, ``Array gain and diversity order of multiuser {MISO}
  antenna systems,'' \emph{International Journal of Wireless Information
  Networks}, vol.~15, no.~2, pp. 82--89, May 2008.

\end{thebibliography}

\end{document}